\documentclass[a4paper,reqno]{amsart}
\usepackage[all]{xy}           
\usepackage{amssymb}           
\usepackage{hyperref}
\usepackage{eucal}
\usepackage{graphicx}
\usepackage{epsfig}
\usepackage[usenames,dvipsnames]{color}
\usepackage{color}
\numberwithin{equation}{section}

\newtheorem{definition}{Definition}[section]
\newtheorem{lemma}[definition]{Lemma}
\newtheorem{theorem}[definition]{Theorem}
\newtheorem{proposition}[definition]{Proposition}
\newtheorem{corollary}[definition]{Corollary}
\newtheorem{remarkth}[definition]{Remark}

\newenvironment{remark}{\begin{remarkth}\upshape}{\hfill$\diamond$\end{remarkth}}
\renewcommand{\emph}[1]{{\bfseries\itshape{#1}}}

\numberwithin{figure}{section}

\newcommand{\R}{\mathbb{R}}      

\newcount\ancho \newcount\anchom \newcount\anchoa
\newcount\anchob \newcount\altura

\newcommand{\ltilde}[3][0]{\altura=0 \advance\altura by #1
           \ancho=#2 \anchom=\ancho \divide\anchom by 2
           \anchoa=\ancho \divide\anchoa by 4
           \anchob=\anchom \advance\anchob by \anchoa
           \kern-3pt \begin{array}[b]{c}
           \begin{picture}(1,1)(\anchom,-\altura)
        \qbezier(0,2)(\anchoa,5)(\anchom,2)
        \qbezier(\anchom,2)(\anchob,-1)(\ancho,4)
        \qbezier(0,2)(\anchoa,4.5)(\anchom,1.8)
        \qbezier(\anchom,1.8)(\anchob,-1.5)(\ancho,4)
       \end{picture} \\[-4pt]{#3}
                       \end{array} \kern-4pt    }

\newcommand{\lhat}[3][0]{\altura=0 \advance\altura by #1
           \ancho=#2 \anchom=\ancho \divide\anchom by 2
           \anchoa=\ancho \divide\anchoa by 4
           \anchob=\anchom \advance\anchob by \anchoa
           \kern-3pt \begin{array}[b]{c}
           \begin{picture}(1,1)(\anchom,-\altura)
        \qbezier(0,2)(\anchoa,4)(\anchom,6)
        \qbezier(\anchom,6)(\anchob,4)(\ancho,2)
        \qbezier(0,2)(\anchoa,3.8)(\anchom,5.6)
        \qbezier(\anchom,5.6)(\anchob,3.8)(\ancho,2)
       \end{picture} \\[-4pt] {#3}
                       \end{array} \kern-4pt    }

\newcommand{\lvec}[1]{\overleftarrow{#1}}
\newcommand{\rvec}[1]{\overrightarrow{#1}}

\newcommand{\I}{I\mkern-7muI}

\makeatletter
\newcommand\prol{\@ifstar{\@proldf}{\@prolpf}}  
\def\@prolpf{\@ifnextchar[{\@prolpf@wrt}{\@prolpf@}}
\def\@prolpf@wrt[#1]#2{\@ifnextchar[{\@prolpf@wrt@at{#1}{#2}}{\@prolpf@wrt@{#1}{#2}}}
\def\@prolpf@wrt@at#1#2[#3]{\prolsymbol^{#1}_{#3}#2}
\def\@prolpf@wrt@#1#2{\prolsymbol^{#1}#2}
\def\@prolpf@#1{\@ifnextchar[{\@prolpf@at{#1}}{\@prolpf@@{#1}}}
\def\@prolpf@at#1[#2]{\prolsymbol_{#2}#1}
\def\@prolpf@@#1{\prolsymbol#1}
\def\@proldf{\@ifnextchar[{\@proldf@wrt}{\@proldf@}}
\def\@proldf@wrt[#1]#2{\@ifnextchar[{\@proldf@wrt@at{#1}{#2}}{\@proldf@wrt@{#1}{#2}}}
\def\@proldf@wrt@at#1#2[#3]{\prolsymbol^{*#1}_{#3}#2}
\def\@proldf@wrt@#1#2{\prolsymbol^{*#1}#2}
\def\@proldf@#1{\@ifnextchar[{\@proldf@at{#1}}{\@proldf@@{#1}}}
\def\@proldf@at#1[#2]{\prolsymbol^*_{#2}#1}
\def\@proldf@@#1{\prolsymbol^*#1}
\def\prolsymbol{\mathcal{T}}
\makeatother

\begin{document}
{\Large

\title[Unimodularity and preservation of volumes in
nonholonomic Mechanics]{Unimodularity and
preservation of volumes in nonholonomic Mechanics}

\author[Y.\ N.\ Fedorov]{Yuri N.\ Fedorov}
\address{Yuri N.\ Fedorov:
Department de Matematica Aplicada I \\
Universitat Politecnica de Catalunya, Barcelona, E-08028 Spain}
\email{Yuri.Fedorov@upc.edu}

\author[L.\ C. \ Garc\'{\i}a-Naranjo]{Luis C. \ Garc\'{\i}a-Naranjo}
\address{L. C.\ Garc\'{\i}a-Naranjo:
Departamento de Matem\'aticas y Mec\'anica \\
IIMAS-UNAM \\
Apdo Postal 20-726,  Mexico City,  01000, Mexico}
\email{luis@mym.iimas.unam.mx}

\author[J.\ C.\ Marrero]{Juan C.\ Marrero}
\address{Juan C.\ Marrero:
ULL-CSIC Geometr\'{\i}a Diferencial y Mec\'anica Geom\'etrica\\
Departamento de Matem\'atica Fundamental, Facultad de
Ma\-te\-m\'a\-ti\-cas, Universidad de La Laguna, La Laguna,
Tenerife, Canary Islands, Spain} \email{jcmarrer@ull.es}

\thanks{This work has been partially supported by MEC (Spain)
Grants MTM2009-13383, MTM2011-15725-E, MTM2012-34478, MTM2012-31714 and the project of the Canary Government ProdID20100210. All the authors are grateful to their institutions for funding our research visits, which allowed the completion of the present article.
}

\keywords{nonholonomic mechanical systems,  linear almost Poisson
structure, Hamiltonian dynamics, kinetic energy Hamiltonians, modular vector field, unimodularity,
invariant volume forms, symmetries, reduction}

\subjclass[2010]{37C40,37J60,70F25,70G45,70G65}

\begin{abstract}
The equations of motion of a mechanical system subjected to nonholonomic linear constraints can be formulated  in terms of a linear almost Poisson structure in a vector bundle. We study the existence of invariant measures for the system
in terms of the unimodularity of this structure. In the presence of symmetries, our approach
allows us to give necessary and sufficient conditions for the existence of an invariant volume,
that unify and improve results existing in the literature. We  present an algorithm to study 
the existence of a smooth invariant volume for nonholonomic mechanical systems with symmetry and we
apply it to several concrete mechanical examples.

\end{abstract}

\maketitle

\tableofcontents

\section{Introduction}

The existence of an invariant measure for a  system of differential equations is a very important property.
From the point of view of dynamical systems, it is a key ingredient for the application of ergodic theory. It is also a crucial hypothesis in Jacobi's theorem of the last multiplier that
establishes integrability of the system via quadratures, see e.g. \cite{arnold}. Moreover, the existence
of a smooth invariant measure imposes certain restrictions on the qualitative nature of the 
fixed points of the system; namely, it prohibits the existence of asymptotic equilibria and limit cycles.
The methods developed in this paper address the general question of the existence of an invariant measure for
nonholonomic mechanical systems.

Suppose that a given system of differential equations is described by a vector field $X$ on an orientable  phase-space manifold $M$
that is equipped with a volume form $\mu$. If $f$ is a positive function on $M$, then $X$ preserves the measure
$f\mu$ if and only if $f$ satisfies {\em Liouville's equation}
\begin{equation}
\label{E:Liouville}
\dot f + f \,\mbox{div}_\mu (X)=0,
\end{equation}
where $\dot f=X(f)$ is the derivative of $f$ along the flow of $X$. In local coordinates, the above equation is
a linear partial differential equation for $f$. 
We remark that the existence of an invariant measure is a global property of the system. 
By the Flow-Box Theorem, Liouville equation can always be solved locally (away from equilibrium points).

{\bf Liouville equation for  nonholonomic systems with symmetry}

Let us summarize the local version of our results. In the case of a nonholonomic mechanical system with linear constraints, the phase-space manifold $M$ is a vector bundle 
$D^*$ over a configuration manifold $Q$. 
The fiber variables $p$ are physically interpreted as momenta.

For simplicity, we restrict our attention to systems without potential forces. We are especially interested in systems that possess   a  symmetry group $G$ that acts on $Q$
leaving the constraints and the kinetic energy function invariant.  In section 
\ref{S:Nonho-Syst} we carry out the details of the reduction of the equations of motion. 
It turns out that the reduced phase-space $D^*/G$  is a vector bundle over the {\em shape space} $\widehat Q:=Q/G$, and the fiber 
variables can be split into $(p_a, p_\alpha)$ where, roughly speaking, the $p_a$  are compatible with the symmetries. 

 The Hamiltonian $H$ of the system is derived from the kinetic energy and  is a positive definite quadratic form, with our choice of $(p_a, p_\alpha)$, is of    block type $H=\frac{1}{2}(\mathcal G^{\alpha \beta}(\hat q^{\hat \iota})p_\alpha p_\beta +\mathcal G^{a b}(\hat q^{\hat \iota})p_a p_b)$, where  $\hat q^{\hat \iota}$ are local coordinates on $\widehat Q$. The reduced equations of motion that 
 obey the Lagrange-D'Alembert principle have the form 
 \begin{equation}
 \label{E:Intro-Ham-Eqns}
 \begin{split}
\frac{d\hat q^{\hat \iota}}{dt}&= \widehat Y_\alpha^{\hat \iota}\frac{\partial H}{\partial p_\alpha} ,\\
\frac{dp_\alpha}{dt}&=-\widehat Y_\alpha ^{\hat \iota}\frac{\partial H}{\partial \hat q^{\hat \iota}}
 - C_{\alpha I}^Jp_J\frac{\partial H}{\partial p_I} ,
 \\
 \frac{dp_a}{dt}&=-C_{a I}^Jp_J\frac{\partial H}{\partial p_I},
 \end{split}
\end{equation}
where the upper-case  indices
$I, J,\dots$ run over the joint range of $a, b, \dots $ and $\alpha, \beta, \dots$,
while
 $\widehat Y_\alpha^{\hat \iota}$  and   $C_{IJ}^K$
are smooth functions of  $\hat q^{\hat \iota}$.  Moreover, the functions $C_{IJ}^K$ are skew-symmetric with respect to the lower indices, i.e. $C_{IJ}^K=-C_{JI}^K$. 

The Lagrangian version of the previous equations was considered in \cite{LeMaMa0} (see also \cite{GrLeMaMa}) for the reduction of a nonholonomic mechanical system. In the particular case when the indices $\hat{\iota}$ run over the same range that the indices $\alpha$, that is, if {\em dimension assumption} holds (see \cite{BKMM}) then the previous equations were considered in \cite{KoMa2}.

As we shall see (consequence of Theorems \ref{T:Basic-Unimodularity} and \ref{T:Main}),
 if there is a  preserved measure $\Phi$ for the system (\ref{E:Intro-Ham-Eqns}) then it may be  chosen of {\em basic type}, i.e.
\begin{equation*}
\Phi=e^{\widehat \sigma (\hat q^{\hat \iota})}\,d\hat q^{\hat \iota} \wedge dp_\alpha \wedge dp_a,
\end{equation*}
for a certain real smooth function $\widehat{\sigma}$ on $\widehat Q$ that needs to be determined. If we enforce the Liouville equation
(\ref{E:Liouville}) (for $f = e^{\widehat{\sigma}}$, $\mu = \Phi$ and $X$ the vector field defined by equations \eqref{E:Intro-Ham-Eqns}),
we get
\begin{equation*}
e^{\widehat \sigma} \left ( \frac{\partial \widehat Y_\alpha^{\hat \iota}}{\partial  \hat q^{\hat \iota}}\frac{\partial H}{\partial p_\alpha}
+  \widehat Y_\alpha^{\hat \iota} \frac{\partial \widehat{\sigma}}{\partial \hat q^{\hat \iota}}\frac{\partial H}{\partial p_\alpha}
-C_{\alpha I}^\alpha \frac{\partial H}{\partial p_I}-C_{aI}^a \frac{\partial H}{\partial p_I}\right )=0,
 \end{equation*}
where we have used the skew-symmetry of the coefficients $C_{IJ}^K$ and the 
equality of mixed partial derivatives of $H$. Rearranging terms we
get
\begin{equation*}
 \left ( \frac{\partial  \widehat  Y_\alpha^{\hat \iota}}{\partial  \hat q^
 {\hat \iota}}
+  \widehat  Y_\alpha^{\hat \iota} \frac{\partial \widehat{\sigma}}{\partial \hat q^{\hat \iota}} + C_{\alpha I}^I\right ) \frac{\partial H}{\partial p_\alpha}
+C_{ a I}^I \frac{\partial H}{\partial p_a}=0.
 \end{equation*}

The form of the Hamiltonian $H$ implies that the latter identity can only hold for all values of $p_a, p_\alpha$ if  the following identities hold separately:
\begin{equation*}
 \frac{\partial  \widehat Y_\alpha^{\hat \iota}}{\partial  \hat q^{\hat \iota}}
+ \widehat  Y_\alpha^{\hat \iota} \frac{\partial \widehat{\sigma}}{\partial \hat q^{\hat \iota}} + C_{\alpha I}^I=0 \quad \mbox{for all $\alpha$,} \qquad \mbox{and} \qquad C_{ a I}^I =0 \quad \mbox{for all $a$}.
 \end{equation*}
This paper is devoted at giving an intrinsic  version of the above conditions and 
 applying them
to several concrete mechanical examples. Notice that only the first condition involves the unknown $\widehat{\sigma}$ and that none of them involves the
Hamiltonian function $H$. All of the relevant information for the existence of an invariant measure is contained in the coefficients $\widehat Y_\alpha^{\hat \iota}$ and $C_{IJ}^K$. 
We shall see that these coefficients determine a 
{\em  linear almost Poisson structure} on $D^*/G$.
The formulation of nonholonomic mechanics in terms of a linear almost Poisson structure
is well-known and has its origins in   \cite{SchMa} (see also \cite{CaLeMa,IbLeMaMa,KoMa2,Marle}).
In this paper we show that
 the above conditions for a preserved measure can 
be formulated intrinsically (Theorem \ref{T:Intrinsic-Cond_Unimd}) in terms of a geometric notion: the {\em unimodularity } of the underlying linear
almost Poisson structure. This formulation was announced in \cite{FeGaMa}.

\medskip
{\bf Unimodularity and measure preservation in mechanics}

The  well known theorem of Liouville states that a Hamiltonian system on a symplectic manifold preserves the
symplectic volume, see e.g. \cite{AbMa}. The situation is not so simple for a Hamiltonian system on a Poisson manifold.
For instance, if the Poisson manifold is the dual space $\mathfrak{g}^*$ of a Lie algebra $\mathfrak{g}$ equipped
with the Lie-Poisson structure, then Kozlov \cite{Ko} showed that the  flow of  a Hamiltonian
of  kinetic type  on $\mathfrak{g}^*$ preserves a smooth measure if and only if the Lie algebra $\mathfrak{g}$
is unimodular.

More generally, a sufficient condition for the existence of a smooth measure for a Hamiltonian system
on an abstract Poisson manifold $M$ is that  the {\em modular class of M} vanishes, 
see e.g. \cite{We}. The modular class is the Poisson cohomology class of the modular vector field of $M$ with respect to a volume form. This
vector field plays an important role in the classification of certain Poisson structures (see \cite{DuHa,GrMaPe,LiXu}). On the other hand, 
the modular class of $M$ is the obstruction to the existence of a duality between the canonical homology and the Poisson
cohomology of $M$ (see \cite{BrZu,EvLuWe,Xu}).

The modular class of a Lie algebroid  was introduced in \cite{EvLuWe} and generalizes the
modular character of a Lie algebra. 
 In Marrero \cite{Marrero} the results of Kozlov \cite{Ko} are generalized
to consider the  preservation of volumes for {\em mechanical} Hamiltonian systems 
 on vector bundles equipped with a linear Poisson structure, which is the dual object
 of a Lie algebroid.
A  Hamiltonian   is  mechanical  if it can be expressed as the sum of the kinetic 
and potential energies. Just as in the case studied by Kozlov, the unimodularity of the Lie algebroid is intimately related
with necessary and sufficient conditions for the existence of an invariant measure.

\medskip

{\bf Contributions and outline of the paper}

In section \ref{S:Unimodularity-General} we introduce the  notion of modular
vector field with respect to a volume form and unimodularity of a general  almost Poisson structure. An almost Poisson structure 
on a manifold $M$ is unimodular if there exists a volume form on $M$ which is preserved by all the Hamiltonian vector fields. 
In section \ref{S:Unimodularity-VecBundle} we specialize to the case when
the almost Poisson structure is linear in a vector bundle and we analyze the relationship between
preserved measures for Hamiltonian vector fields and unimodularity.
The main theoretical result of the paper is contained in Theorem \ref{T:Main}, that 
shows that a kinetic energy Hamiltonian vector field on a vector bundle possesses an
invariant measure if and only if the underlying almost Poisson structure is unimodular. Recently, Grabowski \cite{grabowski}  proved, independently
and using a different approach, a similar result. He introduced the notion of the modular section of a skew-symmetric
algebroid and from this theory he obtained the result.

Section \ref{S:Nonho-Syst} focuses on preservation of volumes for nonholonomic systems.
After carrying out the reduction by a symmetry group from the almost Poisson perspective,
we state and prove Theorem \ref{T:Intrinsic-Cond_Unimd} and Corollary \ref{C:Nonho-Invariant-Measure} that allow us to recover several known results
in the literature  
\cite{CaCoLeMa,Jo,Ko,ZeBo} in a unified framework. Moreover, we are able to express the related 
necessary and sufficient conditions for the existence of an invariant measure
 both in global and local form, and to extend the results to the case where the 
{\em dimension assumption} (see \cite{BKMM}) does not hold.

In section \ref{S:algorithm} we give an algorithm to study the existence of an invariant measure
for the reduction of a nonholonomic system with symmetry, that is applied in section \ref{S:examples}
to several concrete mechanical examples. This allows us to obtain the following original
results:

We prove non-existence of an invariant measure for the reduced equations of motion
of the following problems:
\begin{enumerate}
\item[(i)] As considered in \cite{Vor1, Vor2}, the motion of a rigid body with a planar section,
that is not planar or axially symmetric, and  rolls without slipping over a fixed sphere  (Theorem \ref{T:planar-section-body}).
\item[(ii)] The motion of an inhomogeneous sphere whose center of mass does not
coincide with its geometric center that rolls without slipping on the plane (Chaplygin's top) and is not axially symmetric (Theorem \ref{T:Chaplygin-top}).
\item[(iii)] The rolling without slipping of an inhomogeneous sphere, whose center of mass coincides with
its geometric center,  over/inside a 
circular cylinder (Theorem \ref{T:Ball-on-cylinder}).  
\item[(iv)] We also show that the measure
fails to exist in the limit when the  radius of the cylinder in  (iii)  goes to zero and the given ball is
axially symmetric but not homogeneous (Theorem \ref{T:Ball-on-wire}). This example
utilizes our extended theoretical results since 
the dimension assumption does not hold.
\end{enumerate}

During the referee process of this paper, we applied the methods developed here to show
that the reduced equations for a homogeneous ellipsoid rolling on the plane possess an invariant measure if and only if
the ellipsoid is axially symmetric \cite{GNMarrero}.


\medskip 

\section{Almost Poisson structures and unimodularity}
\label{S:Unimodularity-General}

We begin by introducing the notion of the modular vector field of an almost Poisson structure with
respect to a volume form. This is the first step towards the definition of the concept of unimodularity.

Recall that an {\em almost Poisson structure} on a manifold $M$ is an $\R$-bilinear bracket 
of functions,
\begin{equation*}
\{ \cdot , \cdot \} : C^{\infty}(M)\times C^{\infty}(M) \to  C^{\infty}(M),
\end{equation*}
that satisfies 
\begin{enumerate}
\item[(i)] skew-symmetry: $\{f,g\}=-\{g,f\}$,
\item[(ii)] Leibniz rule:  $\{fh,g\}=f\{h,g\}+h\{f,g\}$,
\end{enumerate}
for all $f,g, h \in C^\infty(M)$. We speak of an {\em almost} Poisson bracket since we do not 
require the Jacobi identity
\begin{equation*}
\{\{f,g\},h\}+ \{\{h,f\},g\}+  \{\{g,h\},f\}=0
\end{equation*}
to hold. If the above identity holds for all functions,
 then $\{ \cdot, \cdot \}$ is a usual Poisson structure on $M$ \cite{Va}.
Our interest in almost Poisson structures comes from nonholonomic mechanics 
 \cite{SchMa} (see also \cite{CaLeMa,IbLeMaMa,KoMa2,Marle}).

As a consequence of Leibniz' rule, for any function $f\in C^\infty(M)$ we can associate a vector
field $X_f$ on $M$, the {\em Hamiltonian vector field}\footnote{Note that $X_f$  is not a Hamiltonian vector field in the usual sense since we work with an almost Poisson bracket. It would be more
accurate to talk about an almost Hamiltonian vector field. We eliminate the ``almost" in our
terminology for brevity.} {\em of $f$}, whose action on $g\in C^\infty (M)$ is defined
by $X_f(g):=\{g,f\}$. Note that by skew-symmetry we have $X_f(g)=-X_g(f)$.

Let $(x^i)$ be a local system of coordinates on $M$, and denote by $\pi_{ij}(x):=\{x^i,x^j\}$. The local
expression for $X_f$ is
\begin{equation*}
X_f(x)=\pi_{ij}(x)\frac{\partial f}{\partial x^j}(x)\frac{\partial }{\partial x^i},
\end{equation*}
and the integral curves of $X_f$ satisfy {\em Hamilton's equations}
\begin{equation*}
\frac{dx^i}{dt}=\pi_{ij}(x)\frac{\partial f}{\partial x^j}(x).
\end{equation*}

Now suppose that $M$ is an orientable manifold and that $\Phi$ is a volume form on $M$. The following proposition is key to the developments in this paper.
\begin{proposition}
The map $\mathcal{M}_\Phi:C^\infty(M)\to C^\infty(M)$ given by
\begin{equation*}
\mathcal{M}_\Phi(f)=\mbox{\em div}_\Phi ( X_f)
\end{equation*}
defines a vector field on $M$, where $\mbox{\em div}_\Phi (X_f)$ denotes the divergence of the
vector field $X_f$ with respect to the volume form $\Phi$.
\end{proposition}
\begin{proof}
We need to show that $\mathcal{M}_\Phi$ is a derivation, i.e. that 
\begin{equation}
\label{E:M_is_derivation}
\mathcal{M}_\Phi(fg)=f\mathcal{M}_\Phi(g)+g\mathcal{M}_\Phi(f).
\end{equation}
We shall see that this is a consequence of the skew-symmetry of the bracket. The reader can 
check that the above identity holds with a simple calculation in local coordinates. Here we
present an intrinsic proof. Denote by $\pounds_X\alpha$ the Lie derivative of the form $\alpha$ along the vector field $X$. Using the properties of the Lie derivative we have:
\begin{equation}
\label{E:Aux_proof_Mod_vectorfield}
\begin{split}
\mathcal{M}_\Phi(fg)\Phi&=\pounds_{X_{fg}}\Phi=\pounds_{fX_{g}}\Phi+\pounds_{gX_{f}}\Phi \\
&=\left ( f\mathcal{M}_\Phi(g)+ g\mathcal{M}_\Phi(f) \right ) \Phi + df\wedge {\bf i}_{X_g}\Phi
+dg\wedge {\bf i}_{X_f}\Phi.
\end{split}
\end{equation}
Using that $df\wedge \Phi=0$, one obtains
\begin{equation*}
0={\bf i}_{X_g}(df\wedge \Phi)=X_g(f)\Phi-df\wedge {\bf i}_{X_g}\Phi,
\end{equation*}
which implies
\begin{equation*}
df\wedge {\bf i}_{X_g}\Phi=X_g(f)\Phi.
\end{equation*}
Similarly, $dg\wedge {\bf i}_{X_f}\Phi=X_f(g)\Phi$. Substituting the latter expressions into \eqref{E:Aux_proof_Mod_vectorfield} and using the skew-symmetry of the bracket yields \eqref{E:M_is_derivation}.
\end{proof}

If $\dim M = m$ and the system of coordinates $(x^i)$ is such that $\Phi=dx^1\wedge\dots\wedge dx^m,$ then the local
expression for $\mathcal{M}_\Phi$ is
\begin{equation}
\label{E:Mod-Vect-Field-Local}
\mathcal{M}_\Phi(x)=\frac{\partial \pi_{ij}}{\partial x^i}(x)\frac{\partial }{\partial x^j}.
\end{equation}

\begin{definition}
The vector field $\mathcal{M}_\Phi$ defined above is called the {\em modular vector field} 
of the almost Poisson structure on $M$ with respect to the volume form $\Phi$. 
\end{definition}

\begin{remark} \label{R:Mod-Class} The notion of modular vector field for a Poisson manifold has been extensively
used \cite{DuHa, GrMaPe, LiXu}. Moreover, Weinstein shows in \cite{We} that it  defines a Poisson 
cohomology class of order 1 which is independent of the choice of volume form. The Poisson manifold
is said to be {\em unimodular} if this cohomology class is zero, or, equivalently, if the modular
vector field is Hamiltonian.
\end{remark}

It is natural to ask how the modular vector field is affected  when we consider a different volume
form $\Phi'=e^\sigma \Phi$ on $M$. Here $\sigma$ is any smooth function on $M$. A straightforward 
calculation gives
\begin{equation}
\label{E:Modular-VF-Diff-Measures}
\mathcal{M}_{\Phi'}=\mathcal{M}_{e^\sigma\Phi}= \mathcal{M}_{\Phi}-X_\sigma.
\end{equation}

The above calculation, together with the comments in Remark \ref{R:Mod-Class}, suggest that
we make the following definition:

\begin{definition} The orientable  almost Poisson manifold $M$ is said to be {\em unimodular}
if there exists a volume form $\Phi$ on $M$ and a real function $\sigma\in C^\infty(M)$ such that 
$\mathcal{M}_\Phi=X_\sigma$.\end{definition}
 
The following Theorem is an immediate consequence of the definition.

\begin{theorem}
\label{T:suff}
If the orientable almost Poisson manifold $M$ is  unimodular  there exists a volume form on $M$ that is invariant
by the flow of all Hamiltonian vector fields.
\end{theorem}
\begin{proof}
If $\mathcal{M}_\Phi=X_\sigma$, then $\mbox{div}_{e^\sigma \Phi}(X_f)=0$ for any function $f\in C^\infty (M)$.
\end{proof}

\section{Linear almost Poisson structures, kinetic energy hamiltonians and unimodularity}
\label{S:Unimodularity-VecBundle}

Theorem  \ref{T:suff} shows that unimodularity is a sufficient condition for the existence of an invariant measure for Hamiltonian vector fields. In this section we  investigate the degree to which this condition is also necessary in the special case when we deal with  linear almost Poisson structures on vector bundles.
We will prove in Theorem \ref{T:Main} that if the Hamiltonian vector field associated to a kinetic energy Hamiltonian preserves a smooth measure then the corresponding almost Poisson structure is
unimodular. This is the main theoretical result of the the paper and will be applied in section 
 \ref{S:examples}
to analyze the existence of an invariant measure of several nonholonomic mechanical systems with symmetry.

\subsection{Linear almost Poisson structures and unimodularity}

Let $\tau: E\to Q$ be a real vector bundle of rank $n$ on a manifold $Q$ of dimension $m$. We shall denote by $E_q$ the fiber over the point $q\in Q$ and by $\Gamma (\tau)$ the space of sections
of $\tau: E\to Q$.

A real function $\varphi:E\to \R$ is said to be {\em linear} if its fiberwise 
restriction $\left . \varphi \right |_{E_q}: E_q\to \R$ is linear for all $q\in Q$. 
On the other hand, a real function  $\tilde \sigma:E\to \R$ is said to be {\em basic} if 
it can be written as $\tilde \sigma=\sigma \circ \tau$ for a certain $\sigma\in C^\infty (Q)$.

It is well known that there exists a one-to-one correspondence between the space of sections
$\Gamma (\tau^*)$ of the dual bundle $\tau^*: E^* \to Q$ and the space of linear functions on $E$.
Such a correspondence is the following: to the section $X\in \Gamma (\tau^*)$, we associate the
linear function $X_\ell:E \to \R$  defined by
\begin{equation*}
X_\ell(\alpha) =X(\tau(\alpha))(\alpha), \qquad \mbox{for} \quad \alpha\in E.
\end{equation*}

For every point $q\in Q$ we can choose local coordinates $(q^i)$ on a chart $ U$
and a local basis $\{e_I\}$  of $ \Gamma (\tau^*)$ on $ U$. We denote by $p_I$ the 
local linear functions: 
\begin{equation*}
p_I=(e_I)_\ell.
\end{equation*}
Then $(q^i, p_I)$ is a system of local bundle coordinates on $\tau^{-1}({U})\subset E$, where $q^{i}$ and $p_I$ are 
basic and linear functions, respectively.

\begin{definition}\cite{LeMaMa0} An almost Poisson structure on a real vector bundle $\tau:E\to Q$
is said to be {\em linear} if the almost Poisson bracket of linear functions is linear.
\end{definition}

The following proposition explains how basic and linear functions behave under the action of linear almost Poisson brackets. Its proof is given in  \cite{LeMaMa0}.

\begin{proposition}\cite{LeMaMa0} Let $\{ \cdot , \cdot \}$ be a linear almost Poisson bracket on $E$.
\begin{enumerate}
\item[(i)] If $X$ is a section of $\tau^*:E^*\to Q$ and $\sigma$ is a smooth function on $Q$, then the bracket
$\{\sigma\circ \tau, X_\ell\}$ is a basic function on $E$ (the bracket of a linear function and a basic function is a basic function).
\item[(ii)]If $\sigma, \sigma'$ are smooth functions on $Q$ then  $\{\sigma \circ \tau, \sigma' \circ \tau\}=0$ (the bracket of
basic functions is zero).
\end{enumerate}
\end{proposition}
The content of the above proposition can be stated in terms of the local bundle coordinates $(q^i, p_I)$ as
\begin{equation}
\{p_I,p_J\}=-C_{IJ}^Kp_K, \qquad \{q^i,p_I\}=\rho_I^i, \qquad \{q^i,q^j\}=0,
\end{equation}
for certain smooth functions $C_{IJ}^K,\, \rho_I^i$ on $Q$. These functions are called the 
{\em local structure functions} of the linear almost Poisson structure. In view of this, we obtain the following expression for the Hamiltonian vector field $X_H$ of a smooth function $H:E \to \R$:
\begin{equation}
\label{E:Local-exp-Ham_VF}
X_H=\left ( \rho_I^i\frac{\partial H}{\partial p_I}\right )\frac{\partial }{\partial q^i} - \left ( \rho_I^i\frac{\partial H}{\partial q^i} +C_{IJ}^Kp_K\frac{\partial H}{\partial p_J} \right )\frac{\partial }{\partial p_I},
\end{equation}
and the corresponding Hamilton's equations
\begin{equation}
\label{E:General-Ham-Eqn}
\frac{dq^i}{dt}=\rho_I^i\frac{\partial H}{\partial p_I},\qquad \frac{dp_I}{dt}= - \left ( \rho_I^i\frac{\partial H}{\partial q^i} +C_{IJ}^Kp_K\frac{\partial H}{\partial p_J} \right ).
\end{equation}
It is hard to underestimate the importance of the above universal form of equations in 
mechanics. The equations of motion of a great number of mechanical systems can
be formulated as above.

Let $\sigma\in C^\infty(Q)$. For further reference 
we note that the Hamiltonian vector field of the basic function
$\sigma \circ \tau:E\to \R$  is the vertical vector field
\begin{equation}
\label{E:Local-basic-Ham-VF}
X_{\sigma \circ \tau}= - \rho_I^i\frac{\partial \sigma}{\partial q^i}\frac{\partial }{\partial p_I}.
\end{equation}

\medskip{{\bf Basic-unimodularity}} 

Suppose that the manifold $Q$ and the vector bundle
$\tau:E\to Q$ are orientable (then the same is true about the dual bundle $\tau^*:E^*\to \R$).
Denote by $\nu$ a volume form on $Q$, and by $\Omega \in \Gamma(\Lambda ^n\tau^*)$
a volume form on the fibers of $E^*$ (that is, $\Omega$ is an everywhere non-vanishing section of the vector bundle $\Lambda^n\tau^*: \Lambda^nE^*\to Q$, where $n$ is the rank
of $E$). As indicated in the Appendix, we can consider {\em a 
basic volume form} $\Phi=\nu\wedge \Omega$ on $E$ (actually any basic volume form on $E$ 
can be obtained through this construction).

Local expressions for the modular vector field $\mathcal{M}_\Phi$ are obtained as follows. Recall that we had selected a 
chart ${U}$ with local coordinates $(q^i)$, and that $\{e_I\}$ was a local basis
of $\Gamma(\tau^*)$ on ${U}$. Suppose moreover that locally we have 
\begin{equation*}
\nu=dq^1\wedge\dots \wedge dq^m, \qquad \Omega=e_1\wedge \dots \wedge e_n.
\end{equation*}
 Then (see Appendix),
\begin{equation*}
\Phi=\nu \wedge \Omega= dq^1\wedge\dots \wedge dq^m\wedge dp_1\wedge \dots \wedge dp_n.
\end{equation*}
Under these assumptions, the local expression \eqref{E:Mod-Vect-Field-Local} for the modular
vector field $\mathcal M_\Phi$ becomes
\begin{equation}
\label{E:Local-basic-Mod-VF}
\mathcal M_{\nu\wedge \Omega}=\left ( \frac{\partial \rho_I^i}{\partial q^i} +C_{IJ}^J  \right )
 \frac{\partial }{\partial p_I} .
\end{equation}
As a consequence we have
\begin{proposition}
\label{P:Basic-Vertical}
The modular vector field of a linear almost Poisson structure on a vector bundle 
$\tau :E\to Q$ with respect
to a basic volume form is $\tau$-vertical.
\end{proposition}

The following definition specializes our interest in the relationship between basic volumes
and unimodularity.

\begin{definition}
Consider a vector bundle $\tau:E\to Q$ equipped with a linear almost Poisson structure.
The structure is said to be {\em basic-unimodular} if there exists a basic volume form $\Phi$
on $E$, and a basic  function $\sigma\circ \tau\in C^\infty (E)$ (with $\sigma \in C^\infty (Q)$)
such that 
$
\mathcal M_\Phi = X_{\sigma \circ \tau}.
$
\end{definition}
\begin{remark}
\label{basic-unimod-independent} Suppose that a linear almost Poisson structure on a vector bundle $\tau: E \to Q$ is basic-unimodular
and that $\Phi'$ is an arbitrary basic volume form on $E$. Then, from (\ref{E:Modular-VF-Diff-Measures}), it follows 
that there exists $\sigma' \in C^{\infty}(Q)$ such that ${\mathcal M}_{\Phi'} = X_{\sigma' \circ \tau}$.
\end{remark}

The following theorem shows that unimodularity, basic-unimodularity, and the fact that all Hamiltonian
vector fields preserve a basic volume are equivalent notions for linear almost Poisson structures on 
vector bundles.
 
\begin{theorem}
\label{T:Basic-Unimodularity}
Let  $\tau:E\to Q$ be a vector bundle equipped with a linear almost Poisson structure.
Suppose that both $Q$ and $E$ are orientable. Then the following statements are equivalent:
\begin{enumerate}
\item[(i)] The almost Poisson structure is unimodular.
\item[(ii)] The almost Poisson structure is basic-unimodular.
\item[(iii)] There exists a basic volume form that is invariant under the flow of all the Hamiltonian vector fields. 
\end{enumerate}

\end{theorem}

\begin{proof}
The implications (iii)$\implies $(ii) and (ii)$\implies $ (i) are trivial.
Moreover, (ii) $\implies $(iii) is a direct consequence of Theorem \ref{T:suff}. Therefore, it only remains
to  prove the 
implication (i) $\implies $(ii). 

Suppose that the linear almost Poisson structure on $E$ is unimodular.
Again, by  Theorem \ref{T:suff}, there exists a volume form $\Phi$ on $E$ that is preserved by
all Hamiltonian vector fields. Such form can be written as
\begin{equation*}
\Phi=e^{\tilde \sigma}\nu \wedge \Omega,
\end{equation*}
where $\nu$ is a volume form on $Q$, $\Omega\in \Gamma(\Lambda^n\tau^*)$ is a volume form
on the fibers of $\tau^*:E^*\to Q$, and $\tilde \sigma$ is a smooth function on $E$. We shall prove that 
$\tilde \sigma$ can be chosen to be basic. First note that for any Hamiltonian function $f\in C^\infty(E)$
we have
\begin{equation*}
0=(\mbox{div}_\Phi X_f) \Phi = \pounds_{X_f}\Phi=\left (X_f(\tilde \sigma) + \mathcal M_{\nu\wedge 
\Omega}(f) \right ) \Phi.
\end{equation*}
Hence, $X_{\tilde \sigma}(f) = \mathcal M_{\nu\wedge 
\Omega}(f)$ for any function $f\in C^\infty (E)$ which implies
\begin{equation}
\label{E:AuxThm}
X_{\tilde \sigma} = \mathcal M_{\nu\wedge \Omega}.
\end{equation}
It follows from Proposition \ref{P:Basic-Vertical} that the vector field $X_{\tilde \sigma}$ is $\tau$-vertical.
Now let $0:Q\to E$ be the zero section of the vector bundle $\tau :E\to Q$. Consider the function
$\sigma:=\tilde \sigma \circ 0\in C^\infty(Q)$. Using the verticality of $X_{\tilde \sigma}$ together with the
 local expression for the Hamiltonian vector fields \eqref{E:Local-exp-Ham_VF} one obtains
 $X_{\tilde \sigma}\circ 0 = X_{ \sigma \circ \tau}\circ 0$. Therefore, evaluation of \eqref{E:AuxThm} along the
 zero section yields
 \begin{equation*}
X_{ \sigma\circ \tau}\circ 0=\mathcal M_{\nu\wedge \Omega}\circ 0.
\end{equation*}
Since both  $X_{ \sigma\circ \tau}$ and $M_{\nu\wedge \Omega}$ are vertical lifts of vector fields
in $Q$, the above identity implies  $X_{ \sigma\circ \tau}=\mathcal M_{\nu\wedge \Omega}$.
In other words, the almost Poisson structure is basic-unimodular. 

\end{proof}

The above theorem, together with the local expressions \eqref{E:Local-basic-Ham-VF}
and \eqref{E:Local-basic-Mod-VF}, give us the following local characterization of unimodularity.
\begin{corollary}
A linear almost Poisson structure on an orientable vector bundle $\tau:E\to Q$ is unimodular
if and only if there exists a smooth function $\sigma:Q\to \R$ such that for any system
of bundle coordinates $(q^i,p_I)$ we have
\begin{equation}
\label{E:Unimodulartiy-Local}
\rho_I^i\frac{\partial \sigma}{\partial q^i}+\frac{\partial \rho_I^i}{\partial q^i} +C_{IJ}^J=0, \qquad
\mbox{for all} \qquad I=1,\dots, n.
\end{equation}

\end{corollary}

\subsection{Kinetic energy Hamiltonians and unimodularity}

Now suppose that the vector bundle $\tau:E\to Q$ is also endowed with a bundle metric
$\mathcal{G}:E\times_QE\to \R$. For each $q\in Q$, we have a scalar product $\mathcal{G}(q):E_q\times E_q\to \R$. In applications to nonholonomic mechanics this metric comes from the kinetic energy of the problem.

The bundle metric $\mathcal{G}$ defines a Hamiltonian function $H$ on $E$. Namely, the {\em kinetic
energy associated to $\mathcal{G}$} that is given by $H(\alpha)=\frac{1}{2}\mathcal{G}(q)(\alpha,\alpha)$
for $\alpha \in E_q$. In terms of our local bundle coordinates $(q^i,p_I)$ we have
\begin{equation}
\label{E:Ham-Kin-Energy}
H(q^i,p_I)=\frac{1}{2}\mathcal{G}^{JK}(q^i)p_Jp_K,
\end{equation}
where the $q$-dependent coefficients $\mathcal{G}^{JK}(q^i)$ are obtained from the local
basis $\{e^I\}$ of $\Gamma(\tau)$ as $\mathcal{G}^{JK}(q):=\mathcal{G}(q)(e^J(q),e^K(q))$. Note that
the condition that $\mathcal{G}$ is a bundle metric implies that the matrix $\mathcal{G}^{JK}(q)$
is symmetric and positive definite for all $q\in Q$. In particular, it is non-degenerate.

We are now ready to state and prove the main theoretical result of the paper that relates the 
existence of an invariant measure for the Hamiltonian vector field of  $H$ defined as above, and the unimodularity of the underlying linear almost Poisson
structure.

\begin{theorem}
\label{T:Main}
Let $\tau:E\to Q$ be a vector bundle equipped with a linear almost Poisson structure. Suppose that 
both the vector bundle and $Q$ are orientable. Let $\mathcal{G}$ be a bundle metric on $E$
and let $H:E\to \R$ be the kinetic energy Hamiltonian associated to   $\mathcal{G}$. Then the
following conditions are equivalent:
\begin{enumerate}
\item[(i)] The almost Poisson structure is unimodular.
\item[(ii)] There exists a volume form $\Phi$ on $E$ that is invariant under the flow of the Hamiltonian vector
field $X_H$. 
\end{enumerate}
\end{theorem}
\begin{proof}

That (i) $\implies $(ii) follows from Theorem \ref{T:suff}. To show the converse fix a volume form
$\nu$ in $Q$ and a volume form $\Omega\in \Gamma(\Lambda^n\tau^*)$ 
  on the fibers of $E^*$. The volume
form $\Phi$ can be written as  $\Phi= e^{\tilde \sigma}\nu\wedge \Omega$ for a certain function $\tilde \sigma \in C^\infty (E)$. Choose local coordinates $(q^i)$ on $Q$ and a local basis $\{e_I\}$ of
$\Gamma(\tau^*)$
satisfying
\begin{equation*}
\nu=dq^1\wedge \dots \wedge dq^m, \qquad \Omega=e_1\wedge \dots \wedge e_n.
\end{equation*}
Denote by $(q^i,p_I)$ the corresponding system of fiber coordinates on $E$. We  then get
\begin{equation*}
\Phi=e^{\tilde \sigma}dq^1\wedge \dots \wedge dq^m\wedge dp_1\wedge \dots \wedge dp_n.
\end{equation*}
Considering the above expression for $\Phi$, together with the local expression \eqref{E:Local-exp-Ham_VF} for the Hamiltonian vector field  $X_H$, then the condition $\mbox{div}_\Phi(X_H)=0$
yields 
\begin{equation*}
\frac{\partial H}{\partial p_I} \left ( \rho_I^i \frac{\partial \tilde \sigma}{\partial q^i}
+\frac{\partial \rho_I^i}{\partial q^i} +C_{IJ}^J\right ) =\frac{\partial \tilde \sigma}{\partial p_I} 
\left (    \rho_I^i \frac{\partial H}{\partial q^i} +C_{IJ}^Kp_K\frac{\partial H}{\partial p_J}\right ).
\end{equation*}
Differentiating this expression with respect to $p_L$ gives
\begin{equation}
\label{E:Aux-Main-Theor-Thm}
\begin{split}
\frac{\partial ^2H}{\partial p_L\partial p_I} &\left ( \rho_I^i \frac{\partial \tilde \sigma}{\partial q^i}
+\frac{\partial \rho_I^i}{\partial q^i} +C_{IJ}^J\right )+\frac{\partial H}{\partial p_I}
 \left ( \rho_I^i \frac{\partial^2 \tilde \sigma}{\partial p_L\partial q^i}
\right ) \\ &=\frac{\partial^2 \tilde \sigma}{\partial p_L\partial p_I} 
\left (    \rho_I^i \frac{\partial H}{\partial q^i} +C_{IJ}^Kp_K\frac{\partial H}{\partial p_J}\right ) 
+\frac{\partial \tilde \sigma}{\partial p_I} 
\left (    \rho_I^i \frac{\partial^2 H}{\partial p_L \partial q^i} +C_{IJ}^L\frac{\partial H}{\partial p_J}
+C_{IJ}^Kp_K\frac{\partial^2 H}{\partial p_L\partial p_J} \right ).
\end{split}
\end{equation}
We now use our hypothesis that the Hamiltonian $H$ is of the form \eqref{E:Ham-Kin-Energy}.
Omitting the dependence on $q^i$ we have
\begin{equation*}
\frac{\partial ^2H}{\partial p_L\partial p_I} =\mathcal{G}^{LI}, \qquad \frac{\partial H}{\partial p_I}=\mathcal{G}^{IJ}p_J, \qquad \frac{\partial H}{\partial q^i}=\frac{1}{2}\frac{\partial \mathcal{G}^{IJ}}
{\partial q^i}p_Ip_J. 
\end{equation*}
Next we evaluate \eqref{E:Aux-Main-Theor-Thm} along the zero section $0:Q\to E$ and
use the above relations to obtain
\begin{equation*}
\mathcal{G}^{LI} \left ( \rho_I^i \left ( \frac{\partial \tilde \sigma}{\partial q^i}\circ 0\right )
+\frac{\partial \rho_I^i}{\partial q^i} +C_{IJ}^J\right )=0.
\end{equation*}
Consider the smooth function $\sigma$ on $Q$ defined as $\sigma=\tilde \sigma \circ 0$.
Then $\displaystyle \frac{\partial \tilde \sigma}{\partial q^i}\circ 0=\frac{\partial  \sigma}{\partial q^i}$. Therefore,
using the non-degeneracy of $\mathcal{G}^{LI}$, 
we obtain
\begin{equation*}
 \rho_I^i \frac{\partial  \sigma}{\partial q^i}
+\frac{\partial \rho_I^i}{\partial q^i} +C_{IJ}^J=0 \qquad \mbox{for all $I=1,\dots, n$},
\end{equation*}
which is the local characterization for unimodularity given in \eqref{E:Unimodulartiy-Local}.
\end{proof}

\section{Invariant measures for nonholonomic mechanical systems with symmetry}
\label{S:Nonho-Syst}

In this section we will layout the almost Poisson reduction of nonholonomic mechanical
systems with symmetry. 
We are particularly careful in the development of the associated geometry, to formulate in Theorem \ref{T:Intrinsic-Cond_Unimd} the necessary and sufficient conditions for the unimodularity of the reduced almost Poisson structure.
This is the key step for our study of preservation of volumes for nonholonomic mechanical
systems. Our result allows us to recover the existing results 
in the area \cite{CaCoLeMa,Jo,Ko,ZeBo} in a unified framework, and to further extend them
to the case where the dimension assumption does not hold.

\subsection{Reduction of nonholonomic mechanical systems with symmetry}
\label{Renomesysy}

We begin by introducing our geometric definition of a symmetric nonholonomic system.

\begin{definition}
Let $p:Q\to \widehat Q:=Q/G$ be a principal $G$-bundle. A symmetric nonholonomic system
is a couple $(\mathcal{G},D)$, where $\mathcal{G}$ is a $G$-invariant Riemannian metric
on $Q$ and $D$ is a $G$-invariant constraint distribution on $Q$.
\end{definition}

Notice that the theory for a general nonholonomic system without symmetries is recovered by supposing that the 
symmetry group is trivial $G=\{e\}$. We will
 describe the dynamics  of the system for a general symmetry group $G$ in terms of linear almost Poisson structures.
As mentioned before,  this formulation goes back to \cite{SchMa} (see also \cite{CaLeMa,IbLeMaMa,KoMa2,Marle}).

\medskip
{\bf The generalized nonholonomic connection}

For our purposes it will be convenient to have a clearer understanding of the structure
of the constraint space $D$ and its relation with the infinitesimal symmetries defined by 
the group action.
Denote by $\mathcal{V}p$ the vertical subbundle of $p$. That is,
\begin{equation*}
\mathcal{V}p(q):=T_q \mbox{Orb}_G(q),
\end{equation*}
where $\mbox{Orb}_G(q)$ denotes the orbit  of $G$ through $q\in Q$.
We will assume that the intersection 
\begin{equation*}
\mathcal{V}^Dp := \mathcal{V}p \cap D
\end{equation*}
has constant rank along $Q$. 
In addition to this hypothesis, a number of important works  (e.g. \cite{BKMM}, \cite{ZeBo}) assume that 
the {\em dimension assumption} holds, namely that
\begin{equation}
\label{E:dim-assum}
 \mathcal{V}p + D = TQ.
\end{equation}
In our approach we will not assume that the latter condition holds in general.
 
Let $ {\mathcal H}:= (\mathcal{V}^Dp)^\perp \cap D$ be the orthogonal complement of $\mathcal{V}^Dp$ in $D$. Under
our assumptions, $ {\mathcal H}$ is a $G$-invariant subbbundle of $D$ that allows us to decompose
\begin{equation}
\label{E:decomp-D}
D=\mathcal{V}^Dp\oplus{\mathcal H}.
\end{equation}

 We define a principal connection, {\em the generalized nonholonomic connection},  whose horizontal distribution is given by
\begin{equation*}
\label{E:HorDist}
\mathfrak{H}:={\mathcal H} \oplus(\mathcal{V}p \oplus {\mathcal H})^{\perp}. 
\end{equation*}
In fact, since $TQ = (\mathcal{V}p \oplus {\mathcal H}) \oplus (\mathcal{V}p \oplus {\mathcal H})^{\perp}$, it follows that
\[
TQ = \mathcal{V}p \oplus \mathfrak{H},
\]
and, using the $G$-invariance of the bundle metric $\mathcal{G}$, one also  checks that the horizontal
 distribution 
 $ \mathfrak{H}$ is $G$-invariant.

If  the dimension assumption \eqref{E:dim-assum} holds, then the horizontal distribution 
 $\mathfrak{H}=\mathcal{H}$, and the principal connection defined above
coincides with the  {\em the nonholonomic connection} introduced in
\cite{BKMM}.

By $G$-invariance of
the vector bundles $\mathcal{V}p$ and $ \mathcal{H}$,  the  decomposition \eqref{E:decomp-D}
drops to the quotient
\begin{equation}
\label{E:Decomp-D-Hat}
\widehat D:=D/G=\mathcal{V}^Dp/G\oplus \mathcal{H}/G.
\end{equation}

The space $\mathcal{V}^Dp/G$ is a vector subbundle (over $\widehat Q$) of $\widehat D$ that we
denote by $\tau_{\mathcal{V}^Dp/G}:\mathcal{V}^Dp/G\to \widehat Q$. The space of sections
$\Gamma(\tau_{\mathcal{V}^Dp/G})$ is naturally identified with the $G$-invariant vector fields $Z$
on $Q$ that lie on the intersection $\mathcal{V}^Dp=D\cap \mathcal{V}p$. As a consequence, 
for any $Z\in \Gamma(\tau_{\mathcal{V}^Dp/G})$, the projection $\widehat Z$ onto $\widehat Q$
via $p$ vanishes. This fact will be used in the proof of Theorem \ref{T:Intrinsic-Cond_Unimd} below.

Now consider the   distribution 
$\widehat{\mathcal H}$ on $\widehat{Q} $ whose characteristic space at the point $\hat{q} \in \widehat{Q}$ is
\[
\widehat{\mathcal H}(\hat{q}) = (T_qp)({\mathcal H}(q)), \; \; \; \mbox { with } q\in Q \mbox{ and } p(q) = \hat{q}.
\]
The above definition of  $\widehat{\mathcal H}$ makes sense since ${\mathcal H}$ is $G$-invariant. Notice that if the dimension assumption \eqref{E:dim-assum} holds, then
$\widehat{\mathcal H}(\hat{q}) =T_{\hat q} \widehat Q$.

On the other hand, the linear map $(T_qp)_{|\mathcal H(q)}: {\mathcal H}(q) \to \widehat{\mathcal H}(\hat{q})$ is an isomorphism of real vector spaces.
In fact, the horizontal lift $^{H}_{q}: \widehat{\mathcal H}(\hat{q}) \to {\mathcal H}(q)$ associated to the generalized nonholonomic connection is just the inverse morphism of $(T_qp)_{|\mathcal H(q)}$. Thus, the horizontal lift of a vector field $\widehat Y$ on $\widehat Q$ that lies on $\widehat{\mathcal H}$ is a $G$-invariant
vector field $\widehat Y^H$ on $Q$ which belongs to the distribution ${\mathcal H}$. That
is, $ \widehat Y^H\in \Gamma (\tau_{\mathcal H/G})$, where 
$\tau_{\mathcal H/G}:\mathcal H/G \to \widehat Q$ is the vector bundle projection.

Therefore, the vector bundles ${\mathcal H}/G \to \widehat{Q} = Q/G$ and $\widehat{\mathcal H} \to \widehat{Q} = Q/G$ are isomorphic and the tangent map
to $p$, $Tp$, induces a vector bundle isomorphism between them.  

\medskip
\pagebreak

{\bf Construction of the reduced almost Poisson bracket}

It is well-known that the space of orbits $T^*Q/G$ of the cotangent lifted action of $G$ to $T^*Q$
is a vector bundle over $\widehat Q$ that admits a linear Poisson structure $\{\cdot , \cdot \}_{T^*Q/G}$,
see e.g. \cite{Ortega-Ratiu}. We shall denote the natural bundle projection by $\hat \tau ^*:
T^*Q/G\to \widehat Q$.

Note that the space of sections of the dual bundle $\hat \tau:TQ/G\to \widehat Q$ may be identified
with the space of $G$-invariant vector fields on $Q$. Let $X$ and $Y$ be any such vector fields and
$\widehat f, \widehat k\in C^\infty (\widehat Q)$. Then the bracket $\{\cdot , \cdot \}_{T^*Q/G}$ is 
characterized by the relations
\begin{equation}
\label{E:bracket-holonomic}
\{ X_\ell , Y_\ell \}_{T^*Q/G}=-\left ( [X,Y] \right )_\ell\, , \qquad 
\{\widehat f \circ \hat \tau^* , X_\ell \}_{T^*Q/G}=\widehat X(\widehat f)\circ \hat \tau^*, \qquad
\{\widehat f \circ \hat \tau^* , \widehat k \circ \hat \tau^* \}_{T^*Q/G}=0,
\end{equation}
where $\widehat X$ denotes the projection of the $G$-invariant vector field $X$ onto $\widehat Q$
via $p$.
In the case where the Lie group $G$ is trivial, the above formulas characterize the canonical Poisson 
structure in $T^*Q$.

Now denote by $\iota^*:T^*Q\to D^*$ and $\mathcal{P}^*:D^*\to T^*Q$ the dual morphisms of the
canonical inclusion $\iota:D\hookrightarrow TQ$, and the $\mathcal{G}$-orthogonal projector
$\mathcal{P}:TQ\to D$. Since $D$ and ${\mathcal G}$ are $G$-invariant, the cotangent lifted action to $T^*Q$ induces
a $G$-action on $D^*$ and the morphisms $\iota^*$ and $\mathcal{P}^*$ are $G$-equivariant.
Thus, they induce the corresponding epimorphism and monomorphism of vector bundles
\begin{equation*}
\widehat \iota^*:T^*Q/G\to \widehat D^*, \qquad \mbox{and}\qquad
\widehat {\mathcal{P}}^*:\widehat D^*\to T^*Q/G
\end{equation*}
where $\widehat D^*:=D^*/G$.

The vector bundle $\tau_{\widehat D^*}:\widehat D^*\to \widehat Q$ possesses a linear almost Poisson structure. 
For functions $\widehat \varphi, \widehat \psi \in C^\infty(\widehat D^*)$ it is given
by
\begin{equation}\label{reduced-lin-almost-Poisson}
\{\widehat \varphi, \widehat \psi \}_{\widehat D^*}=\{ \widehat \varphi \circ \widehat \iota^* , 
\widehat \psi \circ \widehat \iota^* \}_{T^*Q/G}\circ \widehat {\mathcal{P}}^*.
\end{equation}
The bracket $\{\cdot, \cdot\}_{\widehat D^*}$  is the reduced linear almost Poisson structure that describes the reduced nonholonomic dynamics.

We remark that the space of sections of
the dual vector bundle $\tau_{\widehat{D}}:\widehat D\to \widehat Q$ is naturally identified with 
the set of $G$-invariant vector fields on $Q$ that are sections of $D$. Moreover, from (\ref{E:bracket-holonomic}) 
and (\ref{reduced-lin-almost-Poisson}), 
we deduce that the reduced bracket 
$\{\cdot , \cdot \}_{\widehat D^*}$ is characterized by the relations
 \begin{equation}
 \label{E:Nonho-bracket}
\{ X_\ell , Y_\ell \}_{\widehat D^*}=-\left (\mathcal P [X,Y] \right )_\ell\, , \qquad 
\{\widehat f \circ  \tau_{\widehat D^*} , X_\ell \}_{\widehat D^*}=\widehat X(\widehat f)\circ  \tau_{\widehat D^*}, \qquad
\{\widehat f \circ  \tau_{\widehat D^*} , \widehat k \circ  \tau_{\widehat D^*} \}_{\widehat D^*}=0,
\end{equation}
where $X,Y\in \Gamma(\tau_{\widehat D})$, and $\widehat f, \widehat k\in C^\infty (\widehat Q)$
while, as before, $\widehat X$ denotes the projection of $X$ onto $\widehat Q$ via $p$.
The main difference between  the above formulae and \eqref{E:bracket-holonomic}  is
 the appearance of
the projector $\mathcal P$ in the first term.

\subsection{Unimodularity of the reduced linear almost Poisson structure}

The  constructions in the previous section will be useful to study the unimodularity of the reduced almost 
Poisson  structure $\{\cdot , \cdot \}_{\widehat D^*}$. Recall from \eqref{E:decomp-D} that the 
reduced space $\widehat D=D/G$ 
 can be decomposed as $$\widehat D=\mathcal{V}p/G\oplus \mathcal{H}/G.$$ 

Our study of the unimodularity of the almost Poisson structure $\{\cdot , \cdot \}_{\widehat D^*}$ will
be performed by separately considering linear functions induced by sections of the spaces
involved in the above decomposition.

Assume that  both the manifold $\widehat Q$ and the vector bundle $\widehat D^*$
are orientable. Then we can choose a basic volume form $\Phi$ on $\widehat D^*$.

Denote by $\mathcal M_\Phi$ the modular vector field of $\widehat D^*$ with respect to the basic
volume form $\Phi$. We define the family $\mathcal{R}$ of 1-forms $\widehat \omega_\Phi$
on $\widehat Q$ that satisfy
\begin{equation}
\label{E:Def-1-form-om}
\widehat \omega_\Phi(\widehat Y)\circ \tau_{\widehat D^*} =\mathcal M_\Phi((\widehat Y^H)_\ell),
\end{equation}
where $\widehat{Y}$ is any section of $\widehat{ \mathcal{H}}$ and $(\widehat Y^H)_\ell$ denotes the linear
function on $\widehat D^*$ induced by the section $\widehat Y^H\in \Gamma( \tau_{\widehat D})$.
Note that by Proposition \ref{P:Basic-Vertical}, the vector field $\mathcal M_\Phi$ is $\tau_{\widehat D^*} $-vertical. This implies that $\mathcal M_\Phi((\widehat Y^H)_\ell)$ is a basic function with respect
to $\tau_{\widehat D^*} $ and the relation \eqref{E:Def-1-form-om} is consistent.

If the dimension assumption holds, then $\widehat{ \mathcal{H}}=T\widehat Q$ and
$\widehat \omega_\Phi$ is uniquely defined. In other cases, $\mathcal{R}$
is the space of sections of an affine subbundle of $\widehat{TQ}$ with rank 
\[ 
\dim Q - \mbox{rank} (D + \mathcal{V}p).
\]
We are now ready to give necessary and sufficient conditions for 
 the  unimodularity of $\widehat D$ in terms of the family $\mathcal{R}$ of 
 $1$-forms $\widehat{\omega}_{\Phi}$. 

\begin{theorem}
\label{T:Intrinsic-Cond_Unimd}
Let $\Phi$ be any basic volume form on $\widehat D^*$. The linear almost Poisson structure $\{\cdot , \cdot \}_{\widehat D^*}$ on $\widehat D^*$ is 
unimodular if and only if the following conditions hold:
\begin{enumerate}
\item[(i)] For every section $Z$ of $\mathcal{V}^Dp/G$ we have 
\begin{equation*}
\mbox{\em div}_\Phi (X_{Z_\ell})=0,
\end{equation*}
 where $X_{Z_\ell}$ is the Hamiltonian vector field on $\widehat D^*$ associated with
 the linear function $Z_\ell:\widehat D^* \to \R$ induced by the section $Z$.
 \item[(ii)] There exists a 1-form $\widehat \omega_\Phi$ on $\widehat Q$ belonging to the
 family $\mathcal{R}$ defined by \eqref{E:Def-1-form-om} such that
 \begin{equation*}
 \widehat \omega_\Phi=-d\widehat \sigma,
\end{equation*}
for a certain smooth function $\widehat \sigma$ on $\widehat Q$.
\end{enumerate}
\end{theorem}

\begin{proof}
Suppose that the almost Poisson structure is unimodular. From Remark \ref{basic-unimod-independent} 
and Theorem \ref{T:Basic-Unimodularity}, it follows that there exists $\widehat \sigma \in C^\infty(\widehat Q)$ such that
\begin{equation}
\label{E:Aux-Thm-Nonho1}
\mathcal{M}_\Phi=X_{\widehat \sigma\circ  \tau_{\widehat D^*} }.
\end{equation}
Let $Z$ as in item (i) of the statement of the theorem. Then $Z$ is a $G$-invariant vector field
on $Q$ that lies on the intersection $\mathcal{V}^Dp=\mathcal{V}p\cap D$.
Using the definition of the bracket on $\widehat D^*$ given in \eqref{E:Nonho-bracket} we find
\begin{equation*}
X_{\widehat \sigma\circ  \tau_{\widehat D^*}}(Z_\ell)=\{Z_\ell , \widehat \sigma\circ  \tau_{\widehat D^*}
\}_{\widehat D^*}=-\widehat Z(\widehat \sigma) \circ \tau_{\widehat D^*},
\end{equation*}
where $\widehat Z$ is the projection (via $p$) of the vector field $Z$. Thus, since 
$Z$ is $p$-vertical, we have $X_{\widehat \sigma\circ  \tau_{\widehat D^*}}(Z_\ell)=0$.
Therefore, as a consequence of \eqref{E:Aux-Thm-Nonho1} we obtain that the condition
(i) in the theorem is necessary for unimodularity.

Similarly, if $\widehat Y$ is a vector field on $\widehat Q$ that belongs to the distribution
$\widehat{\mathcal H}$ then, since the 1-forms in $\mathcal{R}$ satisfy \eqref{E:Def-1-form-om},
from \eqref{E:Aux-Thm-Nonho1} we obtain
\begin{equation*}
\widehat \omega_\Phi(\widehat Y)\circ \tau_{\widehat D^*} = \mathcal{M}_\Phi ((\widehat Y^H)_\ell)
=X_{\widehat \sigma\circ  \tau_{\widehat D^*}}((\widehat Y^H)_\ell) =\{(\widehat Y^H)_\ell, \widehat \sigma\circ  \tau_{\widehat D^*} \}_{\widehat D^*}=-\widehat Y(\widehat \sigma)\circ \tau_{\widehat D^*} 
\end{equation*}
where we have used again the definition of the bracket on $\widehat D^*$ given in \eqref{E:Nonho-bracket}. The last equality implies that $\widehat \omega_\Phi(\widehat Y)=-\widehat Y(\widehat \sigma)=-d\widehat \sigma(\widehat Y)$.
This shows that item (ii) in the theorem is also a necessary condition for unimodularity.
The sufficiency of these conditions is a consequence of the decomposition \eqref{E:Decomp-D-Hat}.
\end{proof}

\medskip
{\bf Local expressions for unimodularity}

We are interested in obtaining a local version of the above theorem to apply it to concrete examples.
We will use the following notation: $(\hat q^{\hat \iota})$ are local coordinates  on $\widehat Q$,
whereas $(p_a, p_\alpha)$ are linear coordinates on the fibers of $\widehat D^*$ induced by the sections
$\{ s_{I}\} = \{ Z_a, \widehat Y_\alpha^H \}$ that are a $G$-invariant basis of $D$ adapted to the decomposition $D = \mathcal{V}^Dp \oplus
\mathcal{H}$. So,  $\{ s_{I}\} = \{ Z_a, \widehat Y_\alpha^H \}$ induces a basis of $\widehat D$ 
adapted to the decomposition \eqref{E:Decomp-D-Hat}:
\begin{equation*}
\widehat D = \mathcal{V}^Dp/G\oplus \mathcal{H}/G.
\end{equation*}
Here, $\widehat Y_\alpha = \displaystyle \widehat Y_\alpha^{\hat \iota}\frac{\partial}{\partial {\hat q}^{\hat \iota}}$ 
is a local basis of sections of $\widehat{\mathcal{H}}$.  Note the difference between the Latin lower-case indices $a, b, \dots $ and the 
Greek lower-case indices $\alpha, \beta, \dots $. We also remark that the upper-case indices $I, J, \dots $ run over the joint range of $a, b, \dots $ and $\alpha, \beta, \dots $.

Consider the basic volume form $\Phi=\nu \wedge \Lambda$ on $\widehat D^*$ and assume that with our choice of coordinates we have
\begin{equation*}
\nu=\bigwedge_{\hat \iota}d{\hat q}^{\hat \iota},\qquad \Lambda=\bigwedge_{a}Z_a\bigwedge_{\alpha}\widehat Y^H_{\alpha} \ .
\end{equation*}
Then 
\begin{equation}
\label{E:Mesure-local}
\Phi= \bigwedge_{\hat \iota}d{\hat q}^{\hat \iota}\bigwedge_{a}dp_a\bigwedge_{\alpha}dp_\alpha \, .
\end{equation}
On the other hand, using (\ref{E:Local-exp-Ham_VF}) and the fact that $(Z_a)_\ell =p_a$,  $(\widehat Y^H_\alpha )_\ell=p_\alpha$, $\rho^{\hat{\iota}}_a = 0$ and  $\rho^{\hat{\iota}}_\alpha = \widehat{Y}_\alpha^{\hat{\iota}}$, we get
the following expressions for their Hamiltonian vector fields
\begin{eqnarray}
\label{E:X_Za-local}
X_{p_{a}}&=&\left ( C_{ab}^I\,p_I\right )\frac{\partial }{\partial p_b}+\left ( C_{a\alpha}^I\,p_I\right )\frac{\partial }{\partial p_\alpha}\, ,\\
\label{E:X_Yalpha-local}
X_{p_{\alpha}}&=&\widehat Y_\alpha^{\hat \iota}\frac{\partial}{\partial {\hat q}^{\hat \iota}}
+\left ( C_{\alpha a}^I\,p_I\right )\frac{\partial }{\partial p_a}+\left ( C_{\alpha \beta}^I\,p_I\right )\frac{\partial }{\partial p_\beta}\, .
\end{eqnarray}
Note that, using (\ref{E:Nonho-bracket}), we deduce that the local functions $C_{IJ}^K$ on $\widehat Q$ are determined by the equations $\mathcal P [s_I, s_J] = C_{IJ}^K s_K$, where $[\cdot, \cdot]$ is the standard Lie bracket of vector fields on $Q$.

Thus, from  \eqref{E:Mesure-local} and \eqref{E:X_Za-local}, we immediately get
the following local version of item $(i)$ in Theorem \ref{T:Intrinsic-Cond_Unimd}:
\begin{equation}
\label{E:cond1}
C_{aI}^I=C_{ab}^b+C_{a\alpha}^\alpha=0 \qquad \mbox{for all $a$.}
\end{equation}
Similarly, from \eqref{E:Def-1-form-om},  \eqref{E:Mesure-local} and \eqref{E:X_Yalpha-local},
it follows that 
\begin{equation}
\label{E:def-omega-local}
\widehat \omega_\Phi(\widehat Y_\alpha)=\frac{\partial \widehat Y_\alpha^{\hat \iota}}{\partial
\hat q^{\hat i}}+C_{\alpha a}^a+C_{\alpha \beta}^\beta.
\end{equation}
Therefore, we conclude that condition
$(ii)$ in Theorem \ref{T:Intrinsic-Cond_Unimd} is equivalent to the following equations
for the unknown $\widehat \sigma \in C^\infty (\widehat Q)$:
\begin{equation}
\label{E:cond2}
-\widehat Y_\alpha(\widehat \sigma)= \frac{\partial \widehat Y_\alpha^{\hat \iota}}{\partial \hat q^{\hat \iota}}+C_{\alpha b}^b+C_{\alpha \beta}^\beta \qquad \mbox{for all $\alpha$}.
\end{equation}

\subsection{Invariant measures for the reduced nonholonomic system}

We now state our results looking ahead at measure preservation for nonholonomic
systems.

\begin{corollary}
\label{C:Nonho-Invariant-Measure}
The reduced nonholonomic dynamics preserve a volume form if and only if conditions (i)
and (ii) in Theorem \ref{T:Intrinsic-Cond_Unimd} hold for a basic volume form $\Phi$ on
$\widehat D^*$ and a real $C^\infty$-function $\widehat \sigma$ on $\widehat Q$. Furthermore,
if these conditions hold then the dynamics preserves the basic volume form $\Phi'$ on $\widehat D^*$ given by  $\Phi'=\exp({\widehat \sigma\circ \tau_{\widehat D^*}})\,\Phi$.
\end{corollary}
\begin{proof}
The first part follows from Theorems \ref{T:Main} and  \ref{T:Intrinsic-Cond_Unimd}.
On the other hand, if conditions (i) and (ii) in Theorem \ref{T:Intrinsic-Cond_Unimd} hold then the modular vector field $\mathcal{M}_\Phi$ in $\widehat D^*$ is just the Hamiltonian vector 
field of the basic function $\widehat \sigma \circ \tau_{\widehat D^*}$. Thus, using 
\eqref{E:Modular-VF-Diff-Measures}, we deduce that $\mathcal{M}_{\Phi'}=0$, which proves the
second part.
\end{proof}

We will now specialize our discussion to different reduction scenarios.
This will allow us to express  conditions $(i)$ and $(ii)$ in Theorem \ref{T:Intrinsic-Cond_Unimd} 
in simpler terms that are useful in applications. In doing this
 we shall recover several results existing in the literature in a unified manner \cite{CaCoLeMa,Jo, Ko, ZeBo}.

\medskip
{\bf Invariant measures under the dimension assumption}

Suppose that the dimension assumption \eqref{E:dim-assum} holds, that is,
\[
\mathcal{V}p + D = TQ.
\]
Then the horizontal space $\mathfrak{H}$ of the nonholonomic connection equals
$\mathcal{H}$, and the distribution    $\widehat{\mathcal H} = T\widehat{Q}$.
As a direct consequence of Corollary \ref{C:Nonho-Invariant-Measure}
 we recover the
characterization of systems with an invariant measure given in \cite{ZeBo} in a local form.
Here, we present the intrinsic version:
\begin{corollary}\label{Toni-Dmitry}\cite{ZeBo}
The reduced nonholonomic dynamics preserves a volume form if and only if there exists
a basic volume form $\Phi$ on $\widehat D^*$ such that the following   conditions hold:
\begin{enumerate}
 \item
 For every section  $ Z$ of $ \mathcal{V}^Dp/G$ we have
 \[
\mbox{\em div}_{\Phi}(X_{Z_\ell}) = 0. 
 \]
 \item
 The $1$-form $\widehat{\omega}_\Phi$ on $\widehat{Q}$ defined by \eqref{E:Def-1-form-om}
 is exact.
  \end{enumerate}
 Furthermore, if {\em (i)} and {\em (ii)} hold, and $\hat{\omega}_\Phi = -d\hat{\sigma}$, then the dynamics preserves the volume form
  \[
\Phi'=\exp ({\widehat \sigma\circ \tau_{\widehat D^*}})\, \Phi.
\]
 \end{corollary}

The local version of condition (i) in the above corollary is again given by \eqref{E:cond1}.
However, the local version of condition (ii) given in \eqref{E:cond2} greatly simplifies in
this case. Since $\widehat {\mathcal{H}}=T\widehat Q$, the local vector fields 
$\widehat{Y}_{{\alpha}}$ can be chosen as 
\[
\widehat{Y}_{{\alpha}} = \frac{\partial}{\partial \hat{q}^{{\alpha}}},
\]
where $(\hat q^\alpha)$ are local coordinates for $\hat Q$. Equation \eqref{E:def-omega-local}
implies that $\widehat \omega_\Phi$ is locally given by
\begin{equation}
\label{E:localomega-dimassump}
\widehat \omega_\Phi=   \left ( 
 C_{\alpha a}^a + C_{\alpha \beta}^{\beta}
 \, \right ) \, d \hat{q}^{\alpha}.
\end{equation}

\medskip
{\bf Invariant measures for $G$-Chaplygin systems}

Now, suppose that the nonholonomic system is a Chaplygin system, see  e.g. \cite{Koi}, that is,
\[
TQ = \mathcal{V}p \oplus D.
\]
In this case the distribution $\mathcal{V}^Dp = 0$ and  the condition  (i) in Theorem
\ref{T:Intrinsic-Cond_Unimd}
  is empty.
On the other hand, condition (ii) yields the following  result that was proved in \cite{CaCoLeMa}.
\begin{corollary}\label{Cantrijn-Leon-Martin}\cite{CaCoLeMa}
For a $G$-Chaplygin system, the reduced nonholonomic dynamics preserves a volume form if and only if the $1$-form $\widehat{\omega}_\Phi$ on $\widehat{Q}$ is exact. 
Furthermore, if $\widehat \omega_\Phi = -d \widehat \sigma$, with $\widehat \sigma \in C^{\infty}(\widehat Q)$, then the dynamics preserves the volume form
\[
\Phi' = \exp({\widehat \sigma \circ \tau_{\widehat D^*}})\,\Phi.
\] 
\end{corollary} 

In this case, the  local expression \eqref{E:localomega-dimassump} for $\widehat \omega_{\Phi}$ 
further simplifies to give
\begin{equation}
\label{E:Omega-Chaplygin}
\widehat{\omega}_\Phi=  C_{\alpha \beta}^{\beta}  \, d\hat q^\alpha.
\end{equation}

\medskip
{\bf Invariant measures for LL systems}

Now suppose that the nonholonomic system is LL. This means that the configuration manifold coincides with the symmetry Lie group, i.e. $Q = G$.  In this case, both the metric  $\mathcal{G}$ and the constraint distribution $D$ are left-invariant.
Hence they are respectively determined at the Lie algebra ${\frak g}$ by a symmetric, positive definite,  inertia tensor $\I:\mathfrak{g}\to \mathfrak{g}^*$, and a subspace $\partial \subset {\frak g}$ that is not a Lie subalgebra.
For all $h \in G$ we have
\[
\mathcal{G}(h)(u_h, v_h) =  \left \langle \, \I ( (T_hl_{h^{-1}})(u_h)), T_hl_{h^{-1}}(v_h)  \, \right \rangle, \; \; \; D(h) = (T_{\frak e}l_h)(\partial),
\]
where $u_h, v_h\in T_hG$, $\frak{e}$ is the identity element in $G$, $\langle \cdot , \cdot \rangle$ denotes the duality pairing,
and the mapping $l_h:G\to G$ is left multiplication by $h$. 

For an LL system, the condition (ii) in Theorem \ref{T:Intrinsic-Cond_Unimd} is empty.
On the other hand,  the 
local expression \ref{E:cond1} of condition (i) simplifies to:
 \begin{equation*}
C_{ab}^{b}  = 0, \qquad \mbox{ for all } a.
\end{equation*}
Moreover, in this case the vector fields $Z_a$ can be taken as left extensions of a basis 
$\{e_a\}$  of $\partial$ (i.e. $Z_a(h)=T_{\frak e} l_h(e_a)$ for $h\in G$).
Then if we denote  by ${\mathcal P}: {\frak g} \to \partial$ the orthogonal projector corresponding to the metric defined
in ${\frak g}$  by the inertia matrix $\I$, the following result follows
\begin{corollary}\label{linear}
The reduced  nonholonomic dynamics  preserves a volume form if
and only if
\begin{equation}\label{Lie-algebra-linear}
\displaystyle \sum_{b = 1}^{n} \langle \, e^{b}, {\mathcal
P}[e_{a}, e_{b}]_{\frak g} \, \rangle = 0, \; \; \; \mbox{ for  all $ a$,}
\end{equation}
where $\{e_{a}\}$ is a basis of $\partial$ and
$\{e^{a}\}$ is the dual basis of $\partial^*$. In addition,
if (\ref{Lie-algebra-linear}) holds then the nonholonomic dynamics
preserves the euclidean volume on $\partial^*$.
\end{corollary}

The conclusion about the euclidean volume is established by noticing that
any volume of basic type on $\partial^*$ is a constant multiple of it.

%

\begin{remark}{\rm If   the co-dimension of
$\partial$ in $\frak g$  equals to $1$, then the condition 
\eqref{Lie-algebra-linear} can be given in simpler terms. Let $\beta \in {\frak g}^*$ span
the annihilator $\partial^0$ of $\partial$, and let $\eta := \I^{-1}\beta \in {\frak g}$.
%
Then, it is easy to prove that condition
(\ref{Lie-algebra-linear}) is equivalent to the requirement that 
\[
\frac{1}{\langle \beta, \eta \rangle} \mbox{ad}^*_{\eta}\beta + {\mathcal M}_{\frak g} = \mu
\beta,
\]
for a certain $\mu \in \R$, 
where $\mbox{ad}^*: {\frak g} \times {\frak g}^* \to {\frak g}^*$ is the
coadjoint representation of ${\frak g}$, and  ${\mathcal
M}_{\frak g} \in {\frak g}^*$ is the modular character of ${\frak
g}$ (we recall that ${\mathcal M}_{\frak g}(\xi) = \mbox{trace} (\mbox{ad}_{\xi})$, for $\xi \in {\frak g}$). 
This result had been  proven by
Jovanovic \cite{Jo} (see also \cite{Ko}). }
\end{remark}

\section{Algorithm to investigate the existence of invariant measures}
\label{S:algorithm}

In this section we give an algorithm to determine if a nonholonomic system with symmetry
possesses an invariant measure. The algorithm is written so that the conditions in Theorem
\ref{T:Intrinsic-Cond_Unimd} can be applied in practice in a systematic way. We restrict
our attention to a more specific kind of systems that satisfy the following conditions:
\begin{enumerate}

\item[{\bf C1.}] The first de Rham cohomology group of the shape space $\widehat Q=Q/G$ is trivial.

\item[{\bf C2.}] 
 There exists
an open dense set $\widehat U\subset \widehat Q$ with a global chart that defines coordinates $(\hat q^{\hat \iota})$.
\end{enumerate}

For the examples that we treat in the next section, the space
$\widehat Q$ is $\R^2$, $S^2$ or $SO(3)$, that  satisfy both conditions {\bf C1} and  {\bf C2}. We recall
that in our developments we are always assuming that the intersection 
$\mathcal{V}^Dp:=D\cap  \mathcal{V} p$ has constant rank.

The steps are described above under the assumption that $(\mathcal G,D)$ is a symmetric nonholonomic system
on $Q$. Most of  the steps  only involve computations that can be systematically performed with a symbolic algebra program.

\begin{enumerate}
\item[{\bf Step 1.}] Find a basis  $\{Z_a, Y_\alpha \}$ of $G$-invariant  vector fields  of $D$ in such a way
that $\{Z_a \}$ is a basis of $\Gamma(\tau_{\mathcal{V}^Dp})$ and $\{ Y_\alpha \}$ is a basis of $\Gamma(\tau_{\mathcal H})$.

The vector fields $Z_a$ and $Y_\alpha$ need not be defined globally. 
It is sufficient that they are defined on the dense open set $p^{-1}(\widehat U)\subset  Q$.
In the examples treated below, the vector fields  $\{Z_a, Y_\alpha \}$ were found by inspection. 
If the dimension of $ \mathcal{V}^Dp$ is zero, skip steps 2 and 3 and go to step 4. 
On the other hand, if the dimension of $\mathcal H$ is zero, the algorithm terminates in step 3.

\item[{\bf Step 2.}] Compute the structure coefficients $C_{aI}^J$ defined by
\begin{equation*}
\mathcal P \left ( [ Z_a , Z_b ] \right )= C_{a b}^dZ_d + C_{a b}^\alpha Y_\alpha , \qquad 
\mathcal P\left (  [Z_a , Y_\alpha] \right )= C_{a \alpha}^bZ_b + C_{a \alpha}^\beta Y_\beta,
\end{equation*}
where $\mathcal P$ is the $\mathcal{G}$-orthogonal projection onto $D$ and $[\cdot, \cdot]$ is the standard commutator of vector
fields.
Notice that by $G$-invariance of the basis  $\{Z_a, Y_\alpha \}$ and the metric $\mathcal G$,
the structure coefficients $C_{aJ}^K$ are functions of $\hat q^{\hat \iota}$.

\item[{\bf Step 3.}] A necessary condition for the existence of an invariant measure (coming from item (i)
in Theorem \ref{T:Intrinsic-Cond_Unimd} and expressed locally in \eqref{E:cond1}) is that 
\begin{equation*}
C_{a b}^b+C_{a \alpha}^\alpha =0, \qquad \mbox{for all $a$}.
\end{equation*}
If the dimension of $\mathcal H$ is zero, this condition is also sufficient.

\item[{\bf Step 4.}] Compute the coefficients $C_{\alpha \beta}^I$
defined by 
\begin{equation*}
\mathcal P \left ( [Y_\alpha , Y_\beta ] \right ) = C_{ \alpha \beta}^aZ_a + C_{\alpha \beta}^\gamma Y_\gamma.
\end{equation*}
As in step 2, here one needs to compute the standard commutator of vector fields and their
$\mathcal G$-orthogonal projection onto $D$. As before, one should interpret the coefficients $C_{\alpha J}^K$ as functions of $\hat q^{\hat \iota}$.

\item[{\bf Step 5.}] Consider the projections $\widehat Y_{\alpha}$ (via $p$) on $\widehat U$ of the vector fields  $Y_\alpha$, and write them as
\begin{equation*}
\widehat Y_\alpha= \widehat Y_\alpha^{\hat \iota} \frac{\partial}{\partial \hat q^{\hat \iota}}.
\end{equation*}
Now, denote by $\mathcal{R}$ the family of 1-forms $\widehat \omega$ on $\widehat U$ that satisfy
\begin{equation*}
\widehat \omega (\widehat Y_\alpha) = \frac{\partial \widehat Y_\alpha^{\hat \iota} }{\partial \hat q^{\hat \iota}} + C_{\alpha a}^a +
C_{\alpha \beta}^\beta, \qquad \forall \alpha.
\end{equation*}
A necessary condition for the existence of an invariant measure is that there exists
a 1-form $\widehat \omega$ in the family $\mathcal{R}$ that satisfies $\widehat \omega=-d\widehat \sigma_{\widehat U}$ for a certain $\widehat \sigma_{\widehat U} \in C^\infty(\widehat
U)$. In practice, one checks if there is a  $1$-form $\widehat \omega$ within $\mathcal{R}$ that satisfies the weaker condition $d\widehat \omega=0$.

\item[{\bf Step 6}]  Let $(\hat q^{\hat \iota}, p_a,p_\alpha)$ be a system
of local coordinates on $\widehat D^*$ induced by the basis $\{Z_a, Y_\alpha \}$. Then the volume form
\begin{equation*}
\Phi=\exp (\widehat \sigma_{\widehat U})\, \bigwedge_{\hat \iota} d\hat q^{\hat \iota} \bigwedge _a dp_a 
\bigwedge_\alpha dp_\alpha
\end{equation*}
is preserved by the system on the open set $(\tau_{\widehat D^*})^{-1}(\widehat U)\subset \widehat D^*$.
One should finally verify that the above volume form admits an invariant extension to $\widehat D^*$.

\begin{remark}
If the dimension assumption holds, then such an invariant extension always exists. To see this
fix a basic volume form $\Psi$ on $\widehat D^*$,  and suppose that
\begin{equation}
\label{E:Change-vol}
\left . \Psi \right |_{(\tau_{\widehat D^*})^{-1}(\widehat U)}=\exp (\widehat \zeta_{\widehat U} \circ \tau_{ \widehat D^*}) \bigwedge_{\hat \iota} d\hat q^{\hat \iota} \bigwedge _a dp_a 
\bigwedge_\alpha dp_\alpha
, \qquad \mbox{with} \qquad \widehat \zeta_{\widehat U}
\in C^\infty (\widehat U).
\end{equation}
In such a case, there is a smooth extension $\widehat \xi \in C^\infty (\widehat Q)$ of the real function 
$\widehat \sigma_{\widehat U}-\widehat \zeta_{\widehat U}$ and
\begin{equation*}
\Psi'=\exp (\widehat \xi \circ \tau_{ \widehat D^*}) \Psi
\end{equation*}
is an invariant volume form.

In fact, if $\widehat \omega_\Psi$ is the (global) 1-form on $\widehat Q$ given by \eqref{E:Def-1-form-om} then a direct computation using \eqref{E:Change-vol}
proves that
\begin{equation}
\label{E:Change-1-form}
\left . (\widehat \omega_\Psi) \right |_{\widehat U}=\widehat \omega+d\widehat \zeta_{\widehat U}=-d(\widehat \sigma_{\widehat U}
-\widehat \zeta_{\widehat U}).
\end{equation}
In particular, $\widehat \omega_\Psi$ is closed in $\widehat U$ and, since  $\widehat U$ is dense, we obtain that $\widehat \omega_\Psi$
is closed (in $\widehat Q$). This, using {\bf C1}, implies that $\widehat \omega_\Psi$ is exact, that shows that
condition (ii) of Theorem \ref{T:Intrinsic-Cond_Unimd} is satisfied.
\end{remark}

\end{enumerate}

\section{Examples}
\label{S:examples}

\subsection{Body with planar section rolling over a fixed sphere}

Following \cite{BoMaKi,Vor1,Vor2,Yar}, consider the motion of a rigid body that possesses a planar face\footnote{an 
everyday life example of 
such type of rigid body is a shoe (without  heel).} that rolls without slipping over a fixed sphere of 
radius $R$. We consider a body frame $\{E_1, E_2, E_3 \}$, whose origin is located at the center of mass $C$ of the body, and such
that the $E_3$-axis 
is normal to the planar face of the body. The spatial frame   $\{e_1, e_2, e_3 \}$ 
has its origin at
the center $O$ of the sphere. The distance between the center of mass $C$ and
the planar face will be denoted by $\ell$. We assume that our choice of body
frame is such that $E_3$ is the outward normal vector to the sphere at the contact point
$P$ (see Figure \ref{F:esfera}).

The two frames are related by an attitude matrix $g\in SO(3)$.
Let $x\in \R^3$ be the spatial  coordinates of the vector $\rvec{OC}$. 
The body coordinates of $\rvec{OC}$
are $X=g^{-1}x$. At any configuration, the condition 
that the planar face of the body  is tangent to the
sphere is expressed by 
\begin{equation}
\label{E:HolConst}
X_3=R+\ell.
\end{equation}

\begin{figure}[ht]
\centering
\includegraphics[width=10cm]{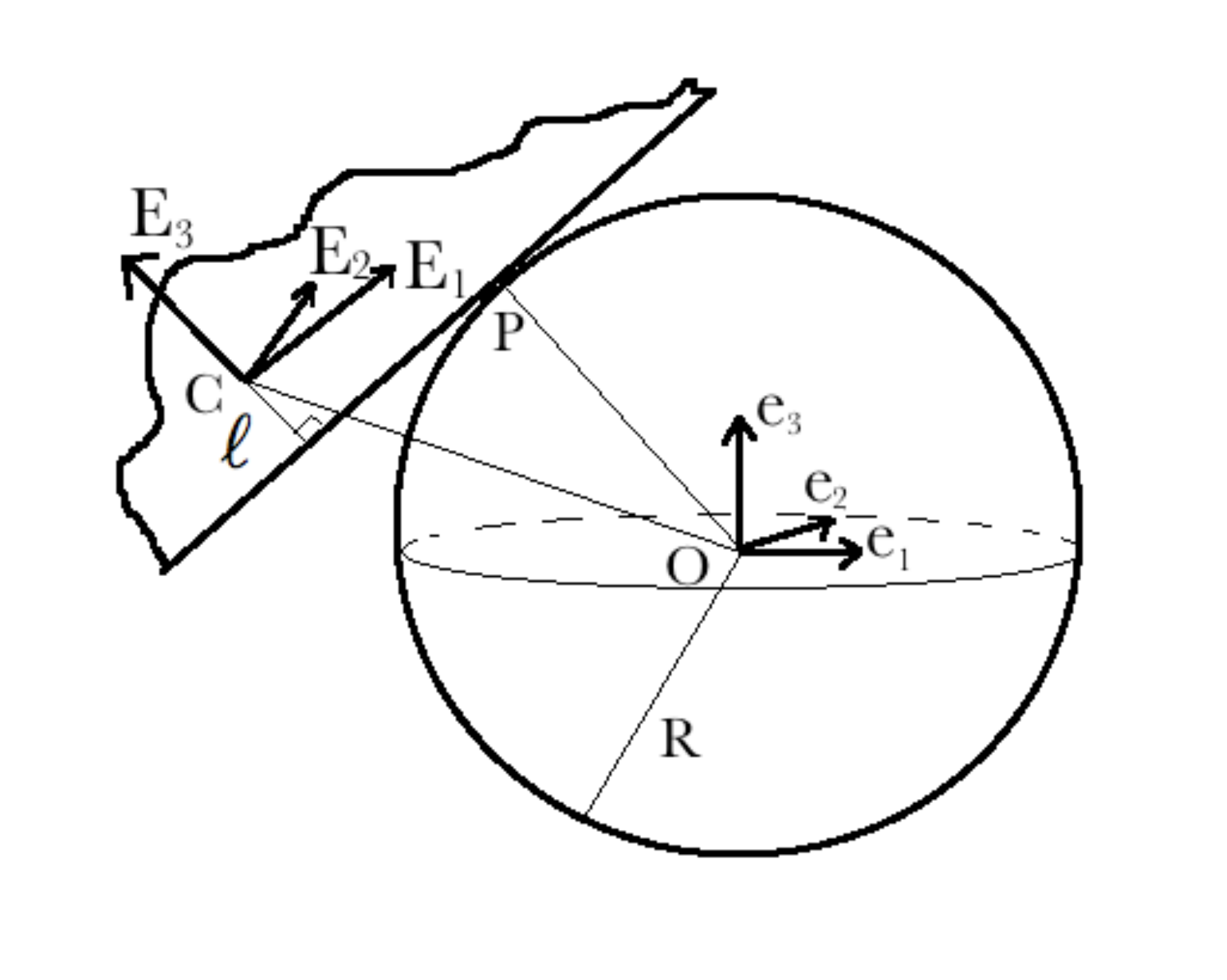}
\caption{\small{Body with planar section rolling over sphere} \label{F:esfera}}
\end{figure}

The constraint that the planar face rolls without slipping is equivalent to the requirement that the
contact point $P$ is at rest. This is expressed by the relations
\begin{equation}
\label{E:Rolling-Const}
\dot X_1 = -R\Omega_2, \qquad \dot X_2 =R\Omega_1,
\end{equation}
where the time derivative of the coordinates $X_1, X_2$ is computed in the 
body frame and the vector $\boldsymbol{\Omega}=(\Omega_1, \Omega_2, \Omega_3)\in \R^3$ is the 
angular velocity of the body expressed in the body frame. We think of $\boldsymbol{\Omega}$
 as an element
in the Lie algebra $so(3)$ corresponding to the skew-symmetric matrix $g^{-1}\dot g$
with the usual identification of $so(3)$ with $\R^3$ via the {\em hat map} (see e.g.
\cite{MaRa}).

In view of our discussion, we can take 
$Q=SO(3)\times \R^2$ with  coordinates $(X_1, X_2)$ for $\R^2$ as the configuration
space. Indeed, if we know the attitude matrix $g$, and the value of the
coordinates $(X_1, X_2)$, then, in view of \eqref{E:HolConst}, we can specify the
configuration of the body. The point of contact $P$ is easily determined since the vector
$\rvec{OP}$ has space coordinates $Rg \tiny{\left ( \begin{array}{c} 0 \\ 0 \\ 1\end{array} \right )} $.

Let $\lvec{e_1}, \lvec{e_2}, \lvec{e_3}$ be the left invariant moving frame  of $SO(3)$ obtained by
left translation of the canonical basis at the group identity. We have the commutation relationships:
\begin{equation*}
[\lvec{e_1}, \lvec{e_2}]= \lvec{e_3}, \qquad [\lvec{e_2}, \lvec{e_3}]= \lvec{e_1}, \qquad [\lvec{e_3}, \lvec{e_1}]= \lvec{e_2}.
\end{equation*}
The entries $\Omega_1, \, \Omega_2, \, \Omega_3$ of the angular velocity in the
body frame $\boldsymbol{\Omega}$ are quasi-velocities with respect to this moving frame.

The no-slip  constraints \eqref{E:Rolling-Const} define  the constraint distribution $D\subset TQ$ that is generated by the vector fields
\begin{equation}
\label{E:Basis-Sphere}
D=\mbox{span} \left \{ \, \frac{\partial}{\partial X_1} - \frac{1}{R}   \lvec{e_2} \, , \,  \frac{\partial}{\partial X_2} + \frac{1}{R}   \lvec{e_1}\, , \,  \lvec{e_3} \right \}.
\end{equation}

The kinetic energy of the body is given by
\begin{equation*}
\mathcal{K}=\frac{1}{2}\langle \I \boldsymbol{\Omega} , \boldsymbol{\Omega} \rangle +\frac{m}{2} || \dot x||^2,
\end{equation*}
where the derivative $\dot x$ is computed on the space frame, $\langle \cdot , \cdot \rangle$ is the euclidean scalar product in $\R^3$, $m$ is the 
mass of the body, and the $3\times 3$ symmetric, positive definite matrix $\I$ with
entries $I_{ij}$ is the inertia
tensor of the body. Notice that in view of our choice
of body axes, $\I$ need not
be diagonal. However, by an appropriate rotation of the axes $E_1, E_2$, we
can assume that $I_{12}=0$.

Using that $g\in SO(3)$ we have
\begin{equation*}
|| \dot x||^2= ||g^{-1} \dot x||^2,
\end{equation*}
and hence, since
\begin{equation*}
g^{-1} \dot x =\dot X + \boldsymbol{\Omega} \times X,
\end{equation*}
we can write
\begin{equation}
\label{E:kinetic-energy}
\mathcal{K}=\frac{1}{2}\langle \I \boldsymbol{\Omega} , \boldsymbol{\Omega} \rangle +\frac{m}{2} || \boldsymbol{\Omega} \times X ||^2 +m \left \langle \dot X , \boldsymbol{\Omega} \times X
\right \rangle +\frac{m}{2}  ||\dot X ||^2.
\end{equation}

Denote by $\{ \lvec{e^1}, \lvec{e^2}, \lvec{e^3} \}$ the dual basis
to  $\{ \lvec{e_1}, \lvec{e_2}, \lvec{e_3} \}$. 
A straightforward calculation using \eqref{E:HolConst} shows that the  kinetic energy $\mathcal{K}$ defines the following
 metric $\mathcal{G}$ in $Q$, 
 \begin{equation}
 \label{E:metric-sphere}
\begin{split}
\mathcal{G}&=(I_{11} + m (X_2^2+(R+\ell)^2))\lvec{e^1}\otimes \lvec{e^1} + (I_{22} + m (X_1^2+(R+\ell)^2))\lvec{e^2}\otimes \lvec{e^2}   \\
& \qquad + (I_{33} + m (X_1^2+X_2^2))\lvec{e^3}\otimes \lvec{e^3} -  2m X_1X_2 \lvec{e^1}\otimes \lvec{e^2} +2(I_{13}-m X_1(R+\ell))
\lvec{e^1}\otimes \lvec{e^3}  \\ & \qquad 
+2(I_{23}-mX_2(R+\ell))\lvec{e^2}\otimes \lvec{e^3} 
+2m(R+\ell) \lvec{e^2}\otimes dX_1 -2mX_2\lvec{e^3}\otimes dX_1  \\ & \qquad + 2mX_1\lvec{e^3}\otimes dX_2  -2m(R+\ell) \lvec{e^1}\otimes dX_2 
+ m ( dX_1 \otimes dX_1 +dX_2 \otimes dX_2).
\end{split}
\end{equation}

The group $G=SO(3)$ acts on $Q$ by left multiplication on the $SO(3)$ factor and leaves the metric $\mathcal{G}$ and the constraint
distribution invariant. 
The   
space $\mathcal{V}^Dp$ has dimension 1 and is generated by 
$ \lvec{e_3}$. The  shape space  $\widehat Q=Q/G= \R^2$ that  has a global chart and its first de
Rham cohomology group is zero. Hence, we are in the framework to apply the algorithm described
in section \ref{S:algorithm}.

The vector fields 
\begin{equation*}
\begin{split}
&Z_{1}=\lvec{e_3}, \qquad Y_{1}=-\frac{1}{R}\lvec{e_2  }+ \left ( \frac{I_{23}-m\ell X_2}{R(I_{33} + m (X_1^2+X_2^2))}\right ) \lvec{e_3} + \frac{\partial}{\partial X_1}\, , \\ 
&Y_{2}=\frac{1}{R}\lvec{e_1} + \left ( \frac{-I_{13}+m\ell X_1}{R(I_{33} + m (X_1^2+X_2^2))} \right ) \lvec{e_3} + \frac{\partial}{\partial X_2} \, ,
\end{split}
\end{equation*}
satisfy the properties of step 1 of the algorithm. To avoid confusion with the use of subindices, in the
treatment of this example,  we denote the
 vector fields  $Z_1, Y_1, Y_2$ respectively by $v_1, v_2, v_3$.
 
To perform step 2, we compute the commutators:
\begin{equation*}
\left [ v_{1} \, , \,  v_{2} \right ]= \left [ Z_{1} \, , \,  Y_{1} \right ] = \frac{1}{R}\lvec{e_1} \, , \qquad 
\left [ v_{1} \, , \,  v_{3} \right ] = \left [ Z_{1} \, , \,  Y_{2} \right ] = \frac{1}{R}\lvec{e_2}.
\end{equation*}
Next, we need to calculate the $\mathcal{G}$-orthogonal projection of these vectors
onto $D$ and express them in terms of the basis $\{ v_1 \, , \,  v_2 \, , \,  v_3 \}$.
This is done by solving the linear system of equations $Ty_J=b_J$, where 
\begin{equation*}
\label{E:Matrix-A}
T=\left ( \begin{array}{ccc} 
\mathcal{G}( v_1,  v_1) & 0 & 0 \\
0 &\mathcal{G}( v_2,  v_2)&\mathcal{G}( v_1,  v_2)
 \\
 0 &\mathcal{G}( v_2,  v_1)&\mathcal{G}( v_2,  v_2) \end{array}
 \right ), \qquad b_J =\left ( \begin{array}{c} 
\mathcal{G}( v_1, \left [ v_1 \, , \,  v_{J} \right ] ) \\
 \mathcal{G}( v_2, \left [ v_1 \, , \,  v_{J} \right ] )  \\
 \mathcal{G}( v_3, \left [ v_1 \, , \,  v_{J} \right ] ) \end{array}
 \right ), \quad J=2,3.
 \end{equation*}
The entries $y_J^K$ of $y_J$ are the structure constants $C_{1, J}^K$.
After solving the systems of equations corresponding to $J=2,3$, with the aid of MAPLE\texttrademark, we obtain:
 \begin{equation*}
C_{1, 2}^{2}+C_{1,3}^{3}=\frac{m}{R^3\det (T)}\left (
 -I_{23}(I_{11} +m\ell^2)X_1 +I_{13}(I_{22} +m\ell^2)X_2 +m\ell (I_{11}-I_{22} )X_1X_2 \right ).
\end{equation*}
According to step 3 of the algorithm, a necessary condition for the existence of an invariant measure
is that the above expression equals zero. Since such equality must hold for all values of $X_1$ and $X_2$, we obtain the
necessary conditions for the existence of an invariant measure:
\begin{equation*}
\begin{split}
\label{E:FirstCondgeneral}
  (I_{11}+m\ell^2)I_{23} =0, \qquad (I_{22}+m\ell^2)I_{13}
=0, \qquad (I_{11}-I_{22})\ell =0.
\end{split}
\end{equation*}
The first two conditions imply that $I_{13} =I_{23}=0$ (i.e.,  the inertia tensor is
diagonal), since $I_{11}, \, I_{22}>0$. The third condition is satisfied if either $\ell =0$
or $I_{11}=I_{22}$. 

The above conditions are also sufficient for the existence of an invariant volume. Indeed,  the measure for the first case was found in \cite{Yar} and for the second case already in
\cite{Vor1} (see also \cite{BMK}), where the equations of motion were also explicitly integrated.

Therefore we have:

\begin{theorem}
\label{T:planar-section-body}
 The reduced equations for a rigid body with a
planar face that rolls without slipping on a fixed  sphere possess an invariant measure if and only if at least one of the following conditions hold:
\begin{enumerate}
\item  the center of mass is in the base plane ($\ell=0$), which actually implies that the body is flat,
\item the body is axially symmetric ($I_{11}=I_{22},\, I_{13} =I_{23}=0$).
\end{enumerate}
\end{theorem}

The preserved measures  have the form
$$
(I_{11} I_{22} + mI_{11}X_1^2+mI_{22}X_2^2)\cdot (I_{33}+ m(X_1^2+X_2^2))^{1/2} \,  
d X_1\wedge dX_2 \wedge d\Omega_1 \wedge d\Omega_2 \wedge d\Omega_3, 
$$
and, respectively,
$$
(I_{11}+m(X_1^2+X_2^2)+m\ell^2)\cdot (I_{11} (I_{33}+m(X_1^2+X_2^2))+m\ell^2 I_{33})^{1/2} \,
d X_1\wedge dX_2\wedge d\Omega_1 \wedge d\Omega_2 \wedge d\Omega_3.
$$

\subsection{The Chaplygin top}

Consider the motion of an inhomogeneous ball that rolls without slipping on the plane. 
If  the center of mass of the sphere coincides with the geometric center,  we recover the well known problem of the Chaplygin sphere, that possesses an invariant measure \cite{chapsphere}.
Here, we consider the general case that is sometimes referred to as the Chaplygin top \cite{Schneider}.
 In the case of an axisymmetric ball, this problem was originally studied in detail in \cite{Routh} and \cite{Ch_r}.

Let $R$ be the radius of the sphere. The space frame $\{e_1, e_2, e_3 \}$ is chosen so that the rolling takes place on the plane
$z=-R$. 
 We consider a body frame $\{E_1, E_2, E_3 \}$, whose origin is located at the  center of mass of the
 sphere $M$, and such that the  geometric center $C$ of the sphere 
lies on the $E_3$-axis.  Let $0\leq \ell \leq R$ be the distance  between the geometric center of the sphere
and the center of mass. The body coordinates of the center $C$ of the sphere are $(0,0,-\ell)$. See Figure
\ref{F:ChapTop} below.

\begin{figure}[ht]
\centering
\includegraphics[width=10cm]{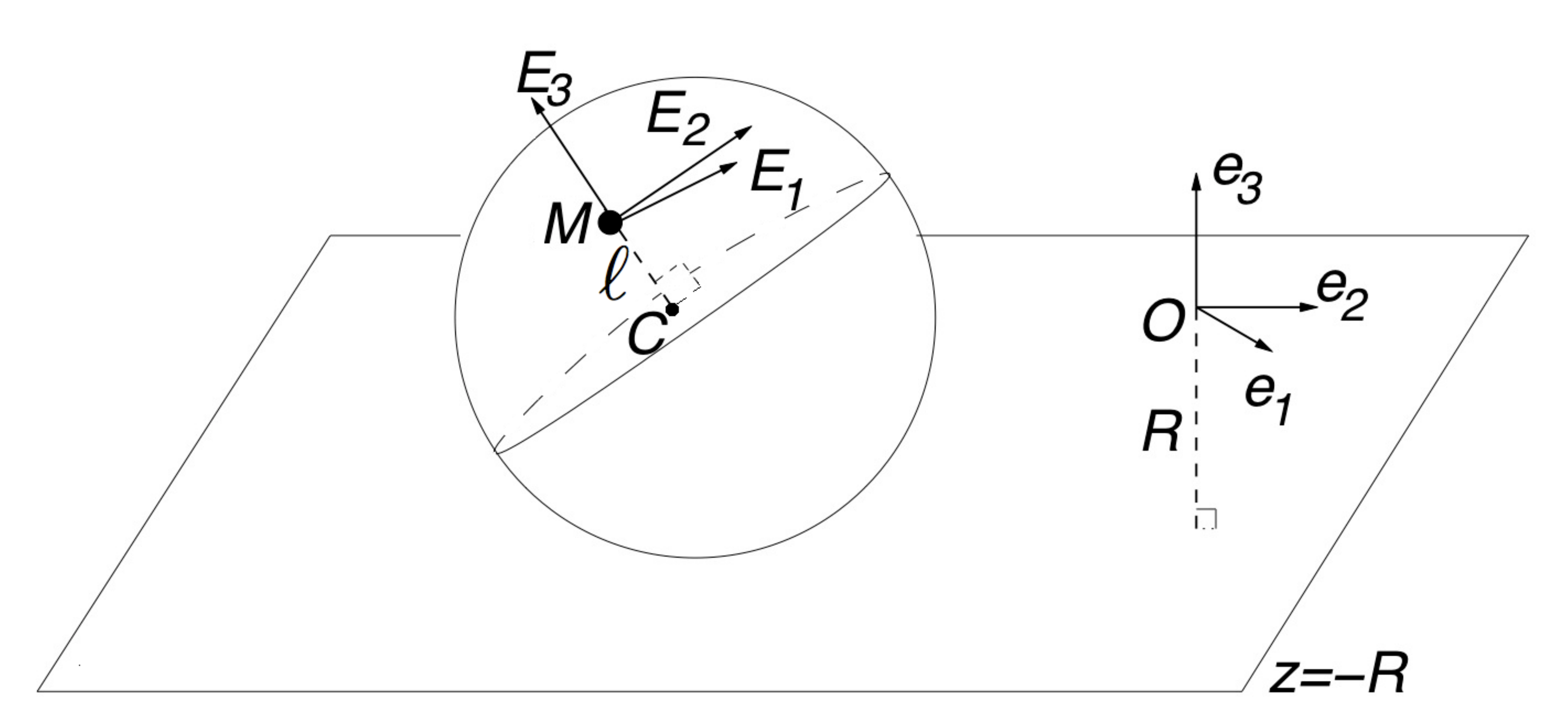}
\caption{\small{Chaplygin Top}\label{F:ChapTop} }
\end{figure}

The configuration space of the system is $Q=SO(3)\times \R^2$. A matrix $g\in SO(3)$ specifies the orientation of the ball by relating the body and the space frame. We will use Euler angles as local
coordinates for 
$SO(3)$. We use  the {\em $x$-convention}, see e.g. \cite{MaRa} and write a matrix $g\in SO(3)$ 
as 
\begin{equation*}
g=\left(
\begin{array}{ccc}
 \cos \psi \cos \varphi - \cos \theta \sin \varphi \sin \psi & -\sin \psi \cos \varphi - \cos \theta \sin \varphi \cos \psi & \sin \theta \sin \varphi \\
\cos \psi \sin \varphi + \cos \theta \cos \varphi  \sin \psi  & -\sin \psi \sin \varphi + \cos \theta \cos \varphi  \cos \psi  & -\sin \theta \cos \varphi  \\
 \sin \theta \sin \psi   & \sin \theta \cos \psi  & \cos \theta  
\end{array}
\right),
\end{equation*}
where the Euler angles $0<\varphi , \psi <2\pi, \, 0<\theta <\pi$. According to this convention, we
obtain the following expressions for the angular velocity in space coordinates $ \boldsymbol{\omega}$,
and in body coordinates $ \boldsymbol{\Omega}$:
\begin{equation}
\label{E:Omega-omega}
 \boldsymbol{\omega}=\left (\begin{array}{c} \dot \theta \cos \varphi +\dot \psi \sin \varphi \sin \theta \\ \dot \theta \sin \varphi - \dot \psi \cos \varphi \sin \theta  \\ \dot 
\varphi + \dot \psi \cos \theta \end{array} \right ), \qquad  \boldsymbol{\Omega}=\left ( \begin{array}{c} \dot \theta \cos \psi +\dot \varphi \sin \psi \sin \theta \\ -\dot \theta \sin \psi + \dot \varphi \cos \psi \sin \theta  \\ \dot 
\varphi\cos \theta + \dot \psi  \end{array} \right ).
\end{equation}
Let $(x,y,0)$ be the spatial coordinates of the geometric center $C$ of the sphere.  The constraints of rolling without slipping are
\begin{equation}
\label{E:Constraints-rolling-Chap-top}
\dot x= R\omega_2 =R (\dot \theta \sin \varphi - \dot \psi \cos \varphi \sin \theta ), \qquad \dot y = -R\omega_1
=-R(\dot \theta \cos \varphi +\dot \psi \sin \varphi \sin \theta).
\end{equation}
The kinetic energy of the sphere is given by
\begin{equation}
\label{E:KinEnergy_ChapTop}
\mathcal{K}=\frac{1}{2}\langle \I \boldsymbol{\Omega} , \boldsymbol{\Omega} \rangle +\frac{m}{2} || \dot {\bf u}||^2,
\end{equation}
where ${\bf u}$ are the space
coordinates of the center of mass. Here, 
 $\I$ is the inertia tensor of the sphere with respect to the center of mass with entries $I_{ij}$ $i,j=1,2,3$. Since we have already made a choice of the body frame, we cannot assume that  $\I$  is diagonal. However, by an
 appropriate rotation of the body frame around the $E_3$-axis, we can assume that $I_{12}=0$. 
 
 The expression for ${\bf u}$  in our coordinates is given by
 \begin{equation}
 \label{E:Def-u}
{\bf u}=\left (  \begin{array}{c} x \\ y \\ 0 \end{array}  \right ) + g \left (  \begin{array}{c} 0 \\ 0 \\ \ell \end{array}  \right ) = \left (  \begin{array}{c} x +\ell  \sin \theta \sin \varphi \\ y -\ell  \sin \theta \cos \varphi\\ \ell \cos \theta \end{array}  \right ) .
\end{equation} 
 Therefore 
  \begin{equation*}
\dot {\bf u}=  \left (  \begin{array}{c} \dot x + \ell\dot \theta  \cos \theta \sin \varphi  + \ell \dot \varphi \sin \theta \cos \varphi  \\ \dot y -\ell \dot \theta  \cos \theta\cos \varphi + \ell \dot \varphi \sin \theta \sin \varphi  \\ -\ell \dot \theta \sin \theta \end{array}  \right ) .
\end{equation*} 
 The local expression for the kinetic energy metric $\mathcal{G}$ of the problem is obtained by substituting the
 expressions for $\dot {\bf u}$ and $\boldsymbol{\Omega}$ onto \eqref{E:KinEnergy_ChapTop}.
 
\medskip{{\bf Symmetries}}

 There is a freedom in the choice of origin and orientation of the space axes $\{ e_1, e_2 \}$. This corresponds to a symmetry of the system defined by a left action of the Euclidean group $SE(2)$ on $Q$. 
 Let 
 \begin{equation*}
h=\left ( \begin{array}{ccc} \cos \vartheta & -\sin \vartheta& v \\  \sin \vartheta & \cos \vartheta& w \\ 0 & 0 & 1 \end{array} \right )
\end{equation*}
denote a generic element on $SE(2)$. The action of $h$ on a point $q\in Q$ with local coordinates
$(\varphi, \theta, \psi, x,y)$ is given by
\begin{equation*}
h\cdot  q = (\varphi + \vartheta ,  \theta, \psi, x \cos\vartheta  - y\sin \vartheta +v , x \sin \vartheta  + y\cos \vartheta +w).
\end{equation*}
One can check that both the constraints and the kinetic energy are invariant under the lift of the action
to $TQ$. The action of $SE(2)$ on $Q$ is free and proper and the shape space $\widehat Q =Q/G =S^2$.
In our local coordinates the orbit projection $p:Q\to S^2$ is given by
\begin{equation*}
p(\varphi, \psi, \theta, x,y) = ( \theta, \psi),
\end{equation*} 
where $( \theta, \psi)$ are spherical coordinates on the unit sphere $\gamma_1^2+\gamma_2^2+\gamma_3^2=1$, defined by
\begin{equation*}
\gamma_1=\sin \theta \sin \psi, \qquad \gamma_2= \sin \theta \cos \psi, \qquad  \gamma_3= \cos \theta.
\end{equation*}
The vector $\boldsymbol{\gamma} = (\gamma_1,\gamma_2,\gamma_3)=g^{-1}e_3$ is the the usual
{\em Poisson vector} whose entries are the  body coordinates of the unit vector that is normal to the plane of
rolling, and evolves according to the kinematic equation
\begin{equation*}
\dot {\boldsymbol{\gamma}}=\boldsymbol{\gamma}\times \boldsymbol{\Omega}.
\end{equation*}

The vertical subbundle $\mathcal{V}p$ is spanned by 
\begin{equation*}
\mathcal{V}p=\mbox{span} \left \{ \frac{\partial}{\partial \varphi}\,  , \,    \frac{\partial}{\partial x} \,  , \,  
 \frac{\partial}{\partial y} \right \}.
\end{equation*}
On the other hand, the constraint distribution $D$ is spanned by the $SE(2)$- invariant vector fields
\begin{equation*}
\begin{split}
Z_{1}&= \frac{\partial}{\partial \varphi}\,  , \\  X_{1}&=\frac{\partial}{\partial \theta} +R\sin \varphi 
 \frac{\partial}{\partial x} - R\cos \varphi 
 \frac{\partial}{\partial y}\,  , \\ X_{2}&=\frac{\partial}{\partial \psi} -R\cos \varphi \sin \theta
 \frac{\partial}{\partial x}-R\sin \varphi \sin \theta
 \frac{\partial}{\partial y}.
 \end{split}
\end{equation*}
It is then clear that the intersection $\mathcal{V}^Dp=D \cap \mathcal{V}p$ has constant dimension 1 and is spanned by $Z_1$. Moreover, since $S^2$ is simply connected, its 
first de Rham cohomology group is zero. Also, the angles $(\theta, \psi)$ form a coordinate chart on an
open dense set of $S^2$, so all of the conditions to apply the algorithm in Section \ref{S:algorithm}
hold.

The following vector fields, together with $Z_{1}$ satisfy the requirements of step 1 of the algorithm:
\begin{equation*}
Y_{\alpha} := X_\alpha - \frac{\mathcal{G}(Z_{1},X_\alpha)}{\mathcal{G}(Z_{1},Z_{1})} \,
Z_{1} \, \qquad \alpha=1, 2.
\end{equation*}
As in our previous example, to avoid confusion with the use of subindices, we denote the
vector fields $Z_1, Y_1, Y_2$ respectively by $v_1, v_2, v_3$.

We compute the standard commutators:
\begin{equation*}
\begin{split}
[v_1, v_2] &= R\cos \varphi  \frac{\partial}{\partial x} + R \sin \varphi  \frac{\partial}
{\partial y} + \lambda_1v_1 \, , \\
[v_1, v_3] &= R\sin \varphi \sin \theta \frac{\partial}{\partial x} - R \cos \varphi \sin \theta \frac{\partial}
{\partial y} + \lambda_2 v_1,
\end{split}
\end{equation*}
where  $\lambda_1, \lambda_2$ are certain functions of $(\theta, \psi)$. We should now
compute the projection of the above vector fields  onto $D$ and express them
as a linear combination of $v_{1}, v_{2}, v_{3}$ to determine the
coefficients $C_{1, J}^K$ with $J=2,3$. In fact,  looking ahead at step 3 of the algorithm,
we are interested in computing $C_{1, J}^J$ for $J=2, 3$.
A simple linear algebra argument shows that $C_{1,2}^{2}$
 coincides with the component of $X_{1}$ when the orthogonal projection
 of 
 \begin{equation*}
R\cos \varphi  \frac{\partial}{\partial x} + R \sin \varphi  \frac{\partial}{\partial y} 
\end{equation*}
onto $D$ is expressed in terms of the basis $Z_{1}, X_{1}, X_{2}$.
 A similar idea can be used to compute $C_{1, 3}^{3}$. Using these observations
 and with the aid of MAPLE\texttrademark\, we obtain:
 \begin{equation}
 \label{E:trace-Chap-Top}
C_{1, 2}^{2}+C_{1, 3}^{3}=\frac{m\ell R \sin^3 \theta}{\mbox{det}
(T)} \left ( a_1(\theta) \cos \psi + a_2(\theta) \cos (2\psi)  + b_1(\theta) \sin \psi + b_2(\theta)
\sin (2\psi) \right )
\end{equation}
where $T$ is the (positive definite) matrix 
\begin{equation*}
T=\left ( \begin{array}{ccc} 
\mathcal{G}( Z_{1},  Z_{1}) & \mathcal{G}( Z_{1},  X_{1}) &  \mathcal{G}( Z_{1},  X_{2}) \\
 \mathcal{G}( Z_{1},  X_{1}) &\mathcal{G}( X_{1},  X_{1})&\mathcal{G}( X_{1},  X_{2})
 \\
  \mathcal{G}( Z_{1},  X_{2}) &\mathcal{G}( X_{1},  X_{2})&\mathcal{G}( X_{2},  X_{2})
 \end{array} \right ),
\end{equation*}
 and the coefficient functions $a_1, a_2, b_1, b_2$ are given by:
 \begin{equation*}
\begin{split}
a_1(\theta)&=-I_{13}\left ( \frac{3m\ell R}{2} + (I_{22} + m(R^2+\ell^2	))\cos (\theta) +  \frac{m\ell R}{2}\cos (2\theta) \right ), \\
a_2(\theta)&=I_{13}I_{23}\sin \theta , \\
b_1(\theta)&=I_{23}\left ( \frac{3m\ell R}{2} + (I_{11} + m(R^2+\ell^2	))\cos (\theta) +  \frac{m\ell R}{2}\cos (2\theta) \right ), \\
b_2(\theta)&= \frac{\sin \theta}{2}\left ( (I_{22}-I_{11})(I_{33}+mR^2+m\ell R\cos \theta)-I_{23}^2+I_{13}^2 \right ).
\end{split}
\end{equation*}

 The necessary condition for the existence of an invariant measure, coming from step 3
 of the algorithm, is that the expression \eqref{E:trace-Chap-Top}
 vanishes identically for all $(\theta, \psi)$ in the chart, that is, for all $0<\theta < \pi$, $0<\psi<2\pi$. By linear independence of $\cos \psi, \cos(2\psi), \sin \psi, \sin(2\psi)$, it follows that this
 can only happen if either $\ell=0$ or all of the coefficient functions $a_1, a_2, b_1, b_2$ vanish.
 A quick examination of the above expressions, shows that the latter case can only hold if
 $I_{23}=I_{13}=0$ and $I_{11}=I_{22}$.
 
 That these conditions are also sufficient is proven in \cite{Yar} (see also \cite{BMK}).
 There it is shown that the reduced equations of motion can be presented in the vectorial form
 \begin{equation}
 \label{E:Motion-Chap-Top}
\dot {\bf K} = {\bf K} \times \boldsymbol{\Omega} + mR(\boldsymbol{\gamma} \times \boldsymbol
{\Omega}) \times ( \boldsymbol{\Omega} \times \boldsymbol{\rho}), \qquad \dot{\boldsymbol{\gamma}}=\boldsymbol{\gamma}\times \boldsymbol{\Omega},
\end{equation}
and possess the invariant measure:
\begin{equation*}
\frac{d\boldsymbol{\gamma} \wedge d{\bf K} }{\sqrt{I_{11}I_{33}+m\langle \boldsymbol{\rho}, \I  \boldsymbol{\rho} \rangle}} = (I_{11}+m|| \boldsymbol{\rho}||^2)\sqrt{I_{11}I_{33}+m\langle \boldsymbol{\rho}, \I  \boldsymbol{\rho} \rangle} \, d\boldsymbol{\gamma} \wedge d\boldsymbol{ \Omega}\,.
\end{equation*}
In the above formulae $ \boldsymbol{\rho}$ is the vector that connects the contact point with the
center of mass of the sphere written in body coordinates, and ${\bf K}$ is the angular momentum of the ball with respect
to the contact point, also written with respect to the body frame. Explicitly we have:
\begin{equation*}
 \boldsymbol{\rho}=R\boldsymbol{\gamma}+ \left ( \begin{array}{c} 0 \\ 0 \\ \ell \end{array} \right ),
 \qquad {\bf K}=\I \boldsymbol{ \Omega} + m \boldsymbol{\rho} \times (\boldsymbol{ \Omega}\times  \boldsymbol{\rho}).
\end{equation*}

 Therefore we have:
\begin{theorem}
\label{T:Chaplygin-top}
The reduced equations for the Chaplygin top possess an invariant measure if and only if at least one of the following two conditions hold:
\begin{enumerate}
\item The center of mass of the sphere coincides with the geometric center ($\ell =0$).
\item The ball is axially symmetric ($I_{11}=I_{22}, I_{13} =I_{23}=0$).
\end{enumerate}
\end{theorem}
We stress that the above conditions were known to be sufficient but we have just shown that they are also
necessary.
\subsection{Dynamically balanced sphere rolling on the exterior/interior of a circular cylinder} 

 We consider the motion of an inhomogeneous sphere, whose center of mass coincides with its geometric center $C$,
that rolls without slipping on the surface of an infinitely long vertical circular cylinder of radius $|r|$. 
If the ball is homogeneous the problem, for general shapes of  cylinders,   was considered in \cite{St}, and is integrable as was shown in \cite{BMK}. We also mention that the problem of a homogeneous ball rolling on   circular cylinder in the presence of gravity was considered by Routh.

The space frame is chosen so that the $z$-axis coincides with the axis of the cylinder. 
Let $\vartheta$ be the 
polar angle of the center of the sphere $C$ on the $xy$ plane. Then, the space coordinates of $C$ are $((R+r)\cos \vartheta, (R+r)\sin\vartheta, z)$ where $R$ is the radius of the
ball.  There are two regions of the parameter $r$ that are physically meaningful. If $r<-R$ this corresponds to 
a sphere rolling on the interior of a circular cylinder of larger radius. If $r>0$, the sphere rolls on the exterior of
a circular cylinder of positive radius.
 Figure
\ref{F:BallonCylinder} below illustrates the case $r>0$.

\begin{figure}[ht]
\centering
\includegraphics[width=5cm]{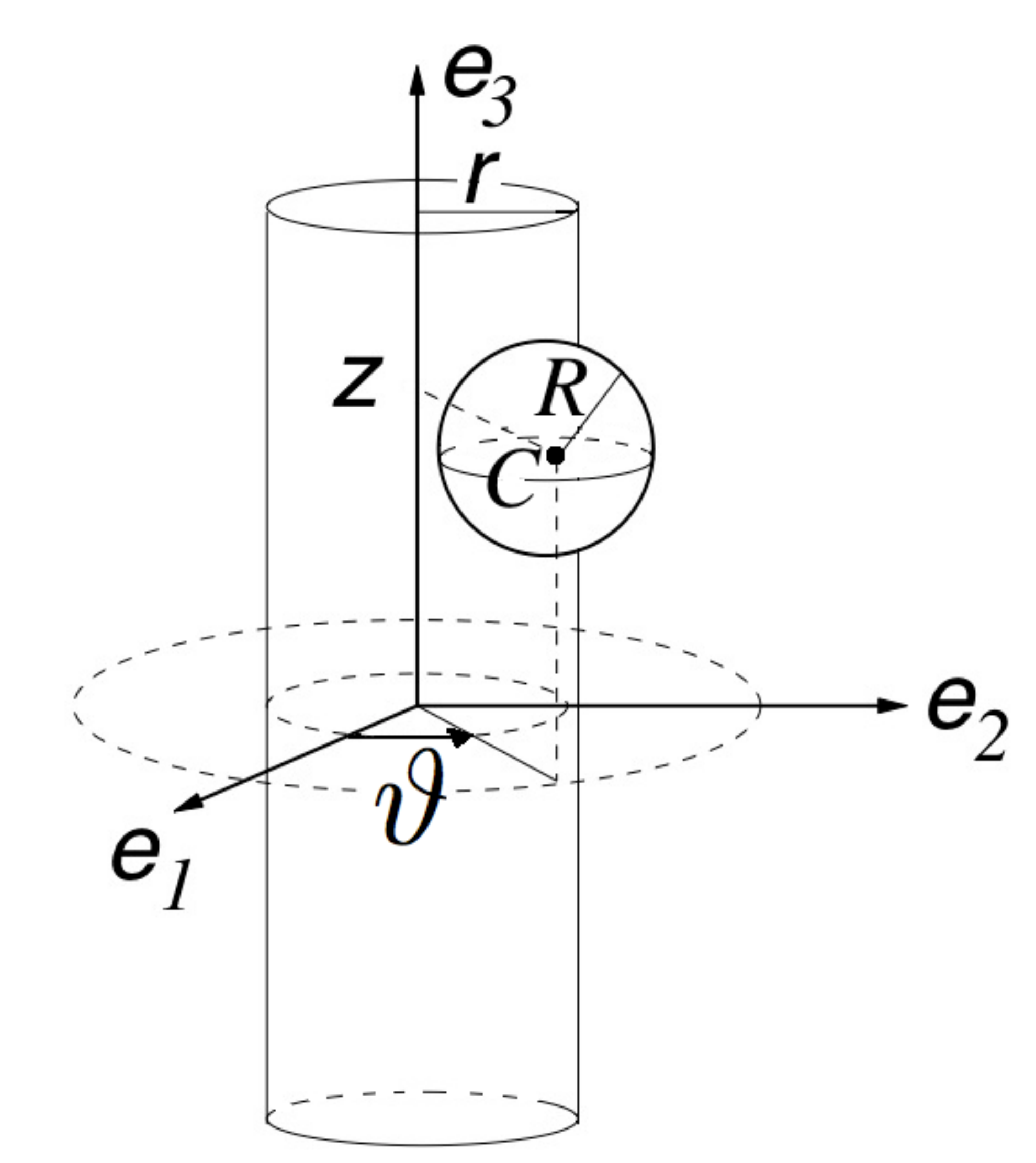}
\caption{\small{Dynamically balanced ball rolling on the exterior of a cylinder ($r>0$).} \label{F:BallonCylinder}}
\end{figure}

 The configuration of the ball is completely determined by the position of its center $C$ and by a rotation
matrix $g\in SO(3)$ that relates the space frame with a fixed body frame, whose origin lies on the center
of the sphere $C$ and that we assume to be aligned with the
principal axes of inertia of the sphere. Hence the configuration space for the problem is
$Q=SO(3)\times \R\times S^1$. Our coordinates for $\R\times S^1$ are $(z, \vartheta)$.

The constraint of rolling without slipping is obtained by requiring that the contact point of the ball
with the cylinder is at rest. This gives
\begin{equation}
\label{E:Const-ball-cylinder}
\dot \vartheta =  \left ( \frac{R}{R+r}\right )\omega_3, \qquad \dot z=R(\omega_1\sin \vartheta -\omega_2 \cos \vartheta),
\end{equation}
where $\boldsymbol{\omega}=(\omega_1,\omega_2,\omega_3)$ is the angular velocity of the sphere written
in space coordinates. In terms of the Euler angle convention introduced in the previous examples,
we have
\begin{equation*}
\dot \vartheta = \left ( \frac{R}{R+r}\right ) \left ( \dot 
\varphi +  \cos \theta \dot \psi  \right ), \qquad \dot z=R\left (- \sin (\varphi- \vartheta)\dot \theta + \sin \theta
\cos  (\varphi- \vartheta) \dot \psi \right ).
\end{equation*}

Since the center of mass of the sphere coincides with the geometric center, the kinetic energy is the
sum of the rotational and translational kinetic energies:
\begin{equation}
\label{E:Kinetic-Energy-ball-cylinder}
K= \frac 12 \langle \I \boldsymbol{\Omega}, \boldsymbol{\Omega} \rangle + \frac{m}{2}\left ( (R+r)^2 \dot \vartheta^2 + \dot z^2 
\right ).
\end{equation}
Here, $\boldsymbol{\Omega}$ denotes the angular velocity written in the body frame and 
$ \I $ is the inertia tensor that is a diagonal $3\times 3$ positive definite matrix $ \I=diag(I_1, I_2,I_3)$ by our choice of body frame.
 
\medskip{{\bf Symmetries}} 

Since both the orientation and the height
of the space frame is arbitrary, we expect to have a $G=\R\times S^1$ symmetry.
The group action on a point $(g,z,\vartheta)$ in the configuration space $Q=SO(3)\times \R \times S^1$ is defined as:
\begin{equation*}
(z',\vartheta ') : (g,z,\vartheta) \mapsto (\frak{R}_{\vartheta '} g,z+z',\vartheta+ \vartheta ')
\end{equation*}
where 
\begin{equation}
\label{E:DefR}
\frak{R}_{\vartheta '}=\left ( \begin{array}{ccc} \cos \vartheta ' & -\sin \vartheta ' & 0\\
\sin \vartheta ' & \cos \vartheta ' & 0 \\ 0 & 0 & 1 \end{array} \right ).
\end{equation}

This is a left action that with our choice of Euler angles is locally given by
\begin{equation*}
(z',\vartheta ') : (\varphi,\theta, \psi,z,\vartheta) \mapsto (\varphi+\vartheta ',\theta, \psi,z+z',\vartheta+ \vartheta ').
\end{equation*}
One can check that both the kinetic energy and the constraints are invariant under the lift
of this action to $TQ$.

The action of $G=\R\times S^1$ on $Q$ is free and proper and the shape space $\widehat Q =Q/G =SO(3)$.
In our local coordinates the orbit projection $p:Q\to SO(3)$ is given by
\begin{equation}
\label{E:Projection-onto-SO(3)}
p(\varphi, \psi, \theta, z,\vartheta) = ( \xi , \psi , \theta),
\end{equation} 
where $\xi =\varphi-\vartheta$, and $( \xi , \psi , \theta)$ are Euler angles for 
$SO(3)$ chosen with respect to
 the $x$-convention.

The vertical subbundle $\mathcal{V}p$ is spanned by 
\begin{equation*}
\mathcal{V}p=\mbox{span} \left \{ \frac{\partial}{\partial \varphi} + \frac{\partial}{\partial \vartheta}\,  , \,    \frac{\partial}{\partial z} \, \right \}.
\end{equation*}
On the other hand, the  3-dimensional constraint distribution $D$ is spanned by the $G$- invariant vector fields
\begin{equation}
\label{E:VF-ball-cylinder}
\begin{split}
Y_1 &=-\frac{\cos \theta \sin (\varphi- \vartheta )}{\sin \theta}\frac{\partial }{\partial \varphi}+\cos (\varphi- \vartheta) \frac{\partial }{\partial \theta} + \frac{ \sin (\varphi- \vartheta)}{\sin \theta}\frac{\partial }{\partial \psi} , \\
Y_2&=R\frac{\partial}{\partial z} - \frac{\cos \theta \cos (\varphi- \vartheta)}{\sin \theta}\frac{\partial }{\partial \varphi}-\sin (\varphi- \vartheta)\frac{\partial }{\partial \theta} +\frac{ \cos (\varphi- \vartheta)}{\sin \theta}\frac{\partial }{\partial \psi}, \\
Y_3&=\frac{\partial }{\partial \varphi}+\left ( \frac{R}{R+r}\right )\frac{\partial }{\partial \vartheta}.
\end{split}
\end{equation}

It is then clear that  $\mathcal{V}p \cap D =\{0\}$ everywhere on $Q$. Therefore we have $TQ=D\oplus \mathcal{V}p$, the dimension assumption is satisfied and we are dealing with
a Chaplygin system. The vector fields $Y_1, Y_2, Y_3$ defined above verify the requirements
of step 1 in the algorithm. Moreover, the condition in step 3 is vacuously satisfied.
We thus concentrate in examining the condition in step 5. 

The projections $\widehat Y_\alpha =Tp (Y_\alpha)$ are given by
\begin{equation*}
\begin{split}
\widehat Y_1 &=-\frac{\cos \theta \sin \xi}{\sin \theta}\frac{\partial }{\partial \xi}+\cos \xi \frac{\partial }{\partial \theta} + \frac{ \sin \xi}{\sin \theta}\frac{\partial }{\partial \psi} , \\
\widehat Y_2&= - \frac{\cos \theta \cos \xi}{\sin \theta}\frac{\partial }{\partial \xi}-\sin \xi\frac{\partial }{\partial \theta} +\frac{ \cos \xi}{\sin \theta}\frac{\partial }{\partial \psi}, \\
\widehat Y_3&= \left ( \frac{r}{R+r}\right )\frac{\partial }{\partial \xi}.
\end{split}
\end{equation*}
 
 Following step 5, we consider the one-form $\widehat \omega$ that satisfies
\begin{equation*}
\begin{split}
\widehat \omega (\widehat Y_1) &=\frac{\partial}{\partial \xi}\left (-\frac{\cos \theta \sin \xi}{\sin \theta} \right ) + \frac{\partial}{\partial \theta} (\cos (\xi)) + \frac{\partial}{\partial \psi} \left ( \frac{ \sin \xi}{\sin \theta} \right ) + f_1(\xi, \theta , \psi) =-\frac{\cos \theta \cos \xi}{\sin \theta} + f_1(\xi, \theta , \psi), 
 \\
\widehat \omega (\widehat Y_2) &=\frac{\partial}{\partial \xi}\left (-\frac{\cos \theta \cos \xi}{\sin \theta} \right ) + \frac{\partial}{\partial \theta} (-\sin (\xi)) + \frac{\partial}{\partial \psi} \left ( \frac{ \cos \xi}{\sin \theta} \right ) + f_2(\xi, \theta , \psi) =\frac{\cos \theta \sin \xi}{\sin \theta} + f_2(\xi, \theta , \psi),\\
\widehat \omega (\widehat Y_3) &=\frac{\partial}{\partial \xi}\left (\frac{r}{R+r} \right ) + f_3(\xi, \theta , \psi) = f_3(\xi, \theta , \psi),
\end{split}
\end{equation*}
where $f_\alpha (\xi, \theta , \psi) = C_{\alpha  \beta}^\beta$ for $\alpha=1,2,3$.
The one-form $\widehat \omega$ is uniquely determined by the above conditions (this is always the case when the dimension assumption holds). Using the explicit expressions for 
$\widehat Y_\alpha$ we find
\begin{equation*}
\begin{split}
\widehat \omega =& \left ( \frac{R+r}{r} \right ) f_3\, d\xi +\left ( - \frac{\cos \theta}{\sin \theta} +
f_1\cos \xi - f_2\sin \xi \right ) \, d\theta \, + \\
& \qquad + \left ( \sin \theta \sin \xi f_1 + \sin \theta \cos \xi f_2 +  \left ( \frac{R+r}{r} \right )\cos \theta f_3 \right ) \, d\psi \, .
\end{split}
\end{equation*}
According to the algorithm, since $(\xi, \theta, \psi)$ are coordinates
in a chart that is open and dense in $SO(3)$, and the first de Rham cohomology group of
$SO(3)$ is trivial, a preserved measure exists if and only if $\widehat \omega$ is closed.

With the aid of a symbolic algebra software one can compute explicit expressions for
$f_\alpha(\xi, \theta, \psi)$. These expressions are too long to be included in the present paper but we can provide the MAPLE\texttrademark \, file upon request. Below we discuss our findings.

 If the sphere is homogeneous ($I_1=I_2=I_3$) one gets
$f_1=f_2=f_3=0$ which implies $\widehat \omega=-d\ln(\sin (\theta))$ and there exists an invariant measure. This result is well known (see e.g. \cite{BMK}). 
We will give an explicit formula for the preserved volume and the reduced equations of
motion at the end of section \ref{S:Ball-on-wire}.

In the general case, the condition that $\widehat \omega$ is closed implies that the functions
\begin{equation*}
\begin{split}
G_1(\xi, \theta , \psi)&:=\left ( \frac{R+r}{r} \right )\frac{\partial f_3}{\partial \theta} -\frac{\partial }{\partial \xi } \left ( f_1\cos \xi - f_2\sin \xi \right ) , \\
G_2(\xi, \theta , \psi)&:=\left ( \frac{R+r}{r} \right )\frac{\partial f_3}{\partial \psi} -\frac{\partial }{\partial \xi }\left ( \sin \theta \sin \xi f_1 + \sin \theta \cos \xi f_2 +  \left ( \frac{R+r}{r} \right )\cos \theta f_3 \right ),
\end{split}
\end{equation*}
vanish identically in the chart $0<\xi<2\pi, \, 0<\theta <\pi, \, 0<\psi <2\pi$. However we found that 
\begin{equation*}
G_1\left (\frac{\pi}{2}, \frac{\pi}{2} , \frac{\pi}{3} \right )=(I_2-I_1) \left ( \frac{R}{R+r} \right ) \left (\frac{\sqrt{3}mR^2}{4((I_1+I_2)mR^2+I_1I_2+(mR^2)^2)} 
\right ),
\end{equation*}
 Hence, a necessary condition for the existence of an invariant measure is
that $I_1=I_2$. Under this assumption one finds:
\begin{equation*}
G_2\left (\frac{\pi}{2}, \frac{\pi}{4} , \frac{\pi}{3} \right )=(I_3-I_1)\left (\frac{\sqrt{2} mR^2(mR^2+\frac{r}{R+r}(I_1+mR^2))}{(mR^2(I_1+I_3)+2I_1I_3)(I_1+mR^2)} \right ),
\end{equation*}
that can only vanish if $I_1=I_3$. Therefore, we have shown
\begin{theorem}
\label{T:Ball-on-cylinder}
The reduced equations for a dynamically balanced sphere that rolls without slipping on the exterior/interior of an infinite circular cylinder possess an invariant measure if and only if the ball is homogeneous ($I_1=I_2=I_3$).
\end{theorem}
 
It is interesting to notice that in the limit when the radius $r$ of the cylinder is infinitely larger than the radius of the sphere one recovers the classical Chaplygin sphere problem that possesses an invariant measure for arbitrary distributions of mass. The other limit case, when the radius of the cylinder $r\to 0$ will be considered in the following section.

To our knowledge, this is the first  mechanical example of a Chaplygin system that has been formally proved not to possess an invariant measure. A mathematical example  had been given in \cite{CaCoLeMa}. 
 
\subsection{Dynamically balanced ball  rolling on a vertical wire}
\label{S:Ball-on-wire}

Consider the motion of an inhomogeneous sphere, whose center of mass coincides with its geometric center,
that rolls without slipping on an infinitely long vertical wire. This corresponds to setting $r= 0$ in the previous section. We follow the same notation and
definitions for $z$ and $\vartheta$.


Putting $r=0$, the constraints of rolling without slipping \eqref{E:Const-ball-cylinder}    become
\begin{equation}
\label{E:Const-ball-on-wire}
\dot \vartheta =  \omega_3, \qquad \dot z=R(\omega_1\sin \vartheta -\omega_2 \cos \vartheta).
\end{equation}
Note that $D$ is still generated by the invariant vector fields $Y_1, Y_2, Y_3$ defined in \eqref{E:VF-ball-cylinder}. However, the problem cannot be treated in the framework of the last section since  the 
vector field $Y_3$ becomes
\begin{equation*}
Z_{1}:=\frac{\partial }{\partial \varphi}+\frac{\partial }{\partial \vartheta},
\end{equation*}
that lies on the vertical subbundle $\mathcal{V}p$. Therefore, in this case, the intersection
$\mathcal{V}^Dp$ has constant rank 1 and we are no longer dealing with a Chaplygin system.
Moreover, notice that 
\begin{equation*}
\dim(D+\mathcal{V}p)=4 < 5=\dim(TQ),
\end{equation*}
at all points $q\in Q$. Hence the dimension assumption is not satisfied.

We follow the steps of the algorithm to analyze the existence of an invariant measure for the reduced equations.
The first thing that we need to do, coming from step 1, is to redefine the vector fields $Y_1, Y_2$ in \eqref{E:VF-ball-cylinder} so that they
are $\mathcal{G}$-orthogonal to $Z_{1}$. For the sake of clearness we change the notation of  the vector fields
in \eqref{E:VF-ball-cylinder}  to
\begin{equation*}
\begin{split}
X_{1} &:=-\frac{\cos \theta \sin (\varphi- \vartheta )}{\sin \theta}\frac{\partial }{\partial \varphi}+\cos (\varphi- \vartheta) \frac{\partial }{\partial \theta} + \frac{ \sin (\varphi- \vartheta)}{\sin \theta}\frac{\partial }{\partial \psi} , \\
X_{2}&:=R\frac{\partial}{\partial z} - \frac{\cos \theta \cos (\varphi- \vartheta)}{\sin \theta}\frac{\partial }{\partial \varphi}-\sin (\varphi- \vartheta)\frac{\partial }{\partial \theta} +\frac{ \cos (\varphi- \vartheta)}{\sin \theta}\frac{\partial }{\partial \psi}, 
\end{split}
\end{equation*}
and define:
\begin{equation*}
Y_\alpha := X_\alpha - \frac{\mathcal{G}(X_\alpha, Z_{1})}{\mathcal{G}(Z_{1}, Z_{1})} \, Z_{1} \, , \qquad \alpha=1,2.
\end{equation*}
Once again, to avoid confusion with the use of subindices, during the treatment of this example we denote
the vector fields $Z_1, Y_1, Y_2$ respectively by $v_1, v_2, v_3$.

For simplicity we only consider the case in which the ball is axially symmetric. We assume that the inertia tensor is
$\I=diag(I_1,I_1,I_3)$. In this case the local expression for the metric tensor $\mathcal{G}$ is
\begin{equation*}
\mathcal{G}=(I_1\sin ^2\theta +I_3\cos^2\theta)\, d\varphi^2 +I_1\, d\theta^2 +I_3\, d\psi^2+ I_3\cos\theta \, d\varphi \otimes
d\psi +mR^2\, d \vartheta^2 +m\, dz^2,
\end{equation*}
and we get
\begin{equation*}
\begin{split}
v_2=Y_1 &=\left( \frac{-(I_3+mR^2) \cos \theta \sin (\varphi- \vartheta )}{\sin \theta (I_1\sin ^2\theta +I_3\cos^2\theta +mR^2)} \right )\frac{\partial }{\partial \varphi}\,+\,\cos (\varphi- \vartheta) \frac{\partial }{\partial \theta} \,+\, \frac{ \sin (\varphi- \vartheta)}{\sin \theta}\frac{\partial }{\partial \psi} + \\ &\qquad \qquad +\,\left (\frac{(I_1-I_3) \cos \theta \sin \theta \sin (\varphi- \vartheta )}{I_1\sin ^2\theta  +I_3\cos^2\theta +mR^2} \right )\frac{\partial }{\partial \vartheta}\,, \\
v_3=Y_{2}&=R\frac{\partial}{\partial z} + \left( \frac{-(I_3+mR^2) \cos \theta \cos (\varphi- \vartheta )}{\sin \theta (I_1\sin ^2\theta +I_3\cos^2\theta +mR^2)} \right )\frac{\partial }{\partial \varphi}\,-\sin (\varphi- \vartheta)\frac{\partial }{\partial \theta} +\frac{ \cos (\varphi- \vartheta)}{\sin \theta}\frac{\partial }{\partial \psi} \, + \\ &\qquad \qquad +\,\left (\frac{(I_1-I_3) \cos \theta \sin \theta \cos (\varphi- \vartheta )}{I_1\sin ^2\theta  +I_3\cos^2\theta +mR^2} \right )\frac{\partial }{\partial \vartheta}\,.
\end{split}
\end{equation*}

It is readily seen that the commutators
\begin{equation*}
[v_1, v_J] =0, \qquad J=1, 2, 3.
\end{equation*}
Therefore, the coefficients $C_{1, J}^K=0$ for all $J, K$, and the condition in step 3 of the algorithm is verified.
We now concentrate in the the study of the condition given in step 5. Recall that the principal 
bundle projection $p:Q\to SO(3)$ is given by \eqref{E:Projection-onto-SO(3)}. The projected vector fields $\widehat v_J =Tp(v_J)$, $J=2,3$, are then
\begin{equation*}
\begin{split}
\widehat v_{2} =-\frac{\cos \theta \sin \xi}{\sin \theta}\frac{\partial }{\partial \xi}+\cos \xi \frac{\partial }{\partial \theta} + \frac{ \sin \xi}{\sin \theta}\frac{\partial }{\partial \psi} , \qquad
\widehat v_{3}= - \frac{\cos \theta \cos \xi}{\sin \theta}\frac{\partial }{\partial \xi}-\sin \xi\frac{\partial }{\partial \theta} +\frac{ \cos \xi}{\sin \theta}\frac{\partial }{\partial \psi}.
\end{split}
\end{equation*}
Following step 5 of the algorithm we consider the family $\mathfrak{R}$ of one-forms $\widehat \omega$
that satisfy
\begin{equation*}
\widehat \omega (\widehat v_{2} )= -\frac{\cos \theta \cos \xi}{\sin \theta} +f_1(\xi, \theta) \, , \qquad \widehat \omega (\widehat v_{3} )= \frac{\cos \theta \sin \xi}{\sin \theta} + f_2(\xi, \theta) \, ,
\end{equation*}
where
\begin{equation*}
f_1(\xi, \theta)= C_{2 , 3}^{3} \, , \qquad        f_2(\xi, \theta) = C_{3 , 2}^{2} \, .
\end{equation*}
The independence of $f_1, f_2$ on $\psi$ is obvious since all objects involved in the calculation are independent
of $\psi$ (see Remark \ref{R:ExtraSymmetry} below). The above relationships imply
that the  family $\mathcal{R}$  can be described as:
\begin{equation}
\label{E:Familyofoneforms}
\widehat \omega =a(\xi, \theta) \, d\theta + 
 b(\xi, \theta) \, d\psi \, +\, \lambda (d\xi + \cos \theta \, d\psi ) \, ,
\end{equation}
where $\lambda =\lambda(\xi, \theta,\psi)$ is arbitrary, and
\begin{equation*}
a(\xi, \theta):=  -\frac{\cos\theta}{\sin \theta} + f_1 \cos \xi - f_2\sin \xi, \qquad b(\xi, \theta):=  f_1\sin \theta\sin \xi + f_2 \sin \theta \cos \xi\, .
\end{equation*}
Step 5 of the algorithm states that a necessary  condition for the existence of an invariant
measure is that there exists a member of this family that is exact.

With the aid of a symbolic algebra software (we can provide the MAPLE\texttrademark \, file upon request), one gets:
\begin{equation*}
f_1(\xi, \theta) =\frac{mR^2I_1(I_3-I_1)}{D} \cos \theta \sin \theta \cos \xi \, , \qquad f_2(\xi, \theta) =\frac{mR^2(I_1+mR^2)(I_1-I_3)}{D} \cos \theta \sin \theta \sin \xi \, ,
\end{equation*}
where $D$ is the determinant of the symmetric positive definite matrix
\begin{equation*}
\left ( \begin{array}{ccc} 
\mathcal{G}( v_{1},  v_{1}) & 0 & 0 \\
0 &\mathcal{G}( v_{2},  v_{2})&\mathcal{G}( v_{2},  v_{3})
 \\
 0 &\mathcal{G}( v_{2},  v_{3})&\mathcal{G}( v_{3},  v_{3}) \end{array}
 \right ). 
\end{equation*}

\medskip{{\bf The homogeneous case}}

 In this case we have  $I_1=I_3$ and hence $f_1=f_2=0$. By taking
$\lambda=0$ we get
\begin{equation*}
\widehat \omega =  -\frac{\cos\theta}{\sin \theta}  \, d\theta \, = -d(\ln(\sin \theta))
\end{equation*}
which is exact. We shall prove that in this case there exists an invariant measure for the problem below by explicitly computing the reduced equations of motion and giving a formula
for the measure.

\medskip{{\bf The inhomogeneous case}} 

We now show that if $I_1\neq I_3$ then the one-form \eqref{E:Familyofoneforms}
is not closed for any smooth function $\lambda$. This implies that no member of $\mathfrak{R}$ is exact and, by step 5 of the algorithm, that there is no preserved measure. The condition $d\widehat \omega=0$ yields the relations
\begin{equation*}
\begin{split}
\frac{\partial \lambda}{\partial \theta}= \frac{\partial a}{\partial \xi} \, ,\qquad \lambda  = \frac{1}{\sin \theta} \left (  \frac{\partial b}{\partial \theta} +\cos \theta  \frac{\partial \lambda}{\partial \theta} \right ) \, ,\qquad \frac{\partial b}{\partial \xi} + \cos \theta \frac{\partial \lambda}{\partial \xi} =  \frac{\partial \lambda}{\partial \psi} \, .
\end{split}
\end{equation*}
Combining the first two relations we get an explicit formula for $\lambda$ that is independent of $\psi$. Substitution into the third
relation gives the following necessary condition for the existence of an invariant measure
\begin{equation*}
\frac{\partial b}{\partial \xi} +  \frac{\cos \theta}{\sin \theta} \left ( \frac{\partial ^2b}{\partial \xi \partial \theta }  + 
\cos \theta \frac{\partial ^2a}{\partial \xi ^2 }
\right ) =0.
\end{equation*}
However, using MAPLE\texttrademark \, we find
\begin{equation*}
\left . \frac{\partial b}{\partial \xi} +  \frac{\cos \theta}{\sin \theta} \left ( \frac{\partial ^2b}{\partial \xi \partial \theta }  + 
\cos \theta \frac{\partial ^2a}{\partial \xi ^2 }
\right ) \, \right |_{(\xi=\pi, \theta = \pi/4)}=-\frac{\sqrt{2}}{4} \left ( \frac{(I_1-I_3)^2(mR^2)^2}{I_1(I_3+mR^2)((mR^2)^2+I_1I_3+mR^2(I_1+I_3)} \right ),
\end{equation*}
that can only vanish if $I_1=I_3$. Therefore, we have shown that there is no preserved measure if $I_1\neq I_3$.

\begin{remark} 
\label{R:ExtraSymmetry}
The assumption that the ball is axially symmetric simplifies the analysis significantly since all
the quantities that appear are independent of the angle $\psi$. This is a consequence of an additional symmetry of
the problem that corresponds to rotations of the sphere about its axis of symmetry. The abelian Lie group $\tilde G= \R\times S^1\times S^1$  acts on the configuration space $Q=SO(3)\times \R\times S^1$ by
the rule
\begin{equation*}
(z',\vartheta ',\psi') : (g,z,\vartheta) \mapsto (\frak{R}_{\vartheta '} g \, \frak{R}_{\psi '},z+z',\vartheta+ \vartheta ')
\end{equation*}
where the matrices $\frak{R}_{\vartheta '} \, , \frak{R}_{\psi '}$ are defined as in \eqref{E:DefR}. This action is
free and proper, and leaves the constraints and the kinetic energy invariant. Nevertheless, the methods described in this paper and in \cite{ZeBo} cannot be used to analyze the existence of an invariant measure of the resulting reduced equations. The reason for this is that the rank of the intersection $\mathcal{V}^Dp=\mathcal{V}p\cap D$ is not 
constant throughout $Q$. If the axis of symmetry of the sphere is perpendicular to the vertical wire, then the rank 
of $\mathcal{V}^Dp$ is 2. 
In any other configuration this rank is 1. Therefore, we are forced to work with the action of the smaller symmetry
group $G=\R\times S^1$ (for which the dimension assumption is not satisfied).
\end{remark}

\medskip{{\bf The reduced equations and the expression for the invariant measure for a homogeneous sphere that rolls without slipping on a vertical cylinder and wire}}

The expression for the invariant measure for a homogeneous ball that rolls without slipping on a vertical cylinder or wire can be readily obtained
as we now show. We treat both problems simultaneously by considering $r<-R$ or $r\geq 0$ where,
as before,  $|r|$ denotes the radius of the cylinder (the case of the wire corresponds to 
the case $r=0$). We begin by introducing the {\em modified Poisson vectors}
\begin{eqnarray*}
\boldsymbol{\alpha}:=\cos \vartheta g^{-1}e_1 +\sin  \vartheta g^{-1}e_2, \qquad 
\boldsymbol{\beta}:=-\sin \vartheta g^{-1}e_1 +\cos  \vartheta g^{-1}e_2 , \qquad 
\boldsymbol{\gamma}:= g^{-1}e_3.
\end{eqnarray*}
They form an orthonormal basis of $\R^3$ and can be considered as the columns of an element in 
$SO(3)$. The constraints \eqref{E:Const-ball-cylinder} can be written in terms of these vectors
as
\begin{equation}
\label{E:Const-ball-on-wire-2}
\dot \vartheta = \left ( \frac{R}{R+r} \right ) \langle  \boldsymbol{\Omega} ,\boldsymbol{\gamma} \rangle \qquad \dot z=-R
 \langle  \boldsymbol{\Omega} ,\boldsymbol{\beta} \rangle .
 \end{equation}
Using the  constraints one can obtain the following kinematical evolution equations for the Poisson vectors\begin{equation}
\begin{split}
\label{E:Evolution-Poisson}
\dot{\boldsymbol{\alpha}}&=\left ( \frac{R}{R+r} \right ) \langle \boldsymbol{\gamma}, \boldsymbol{\Omega} \rangle \boldsymbol{\beta} + \boldsymbol{\alpha} \times \boldsymbol{\Omega} ,\\
\dot{\boldsymbol{\beta}}&=-\left ( \frac{R}{R+r} \right )\langle \boldsymbol{\gamma}, \boldsymbol{\Omega} \rangle \boldsymbol{\alpha} + \boldsymbol{\beta} \times \boldsymbol{\Omega} , \\
\dot{\boldsymbol{\gamma}}&= \boldsymbol{\gamma} \times \boldsymbol{\Omega} .
\end{split}
\end{equation}
Recall that  the Lagrangian of the system \eqref{E:Kinetic-Energy-ball-cylinder}
is given by
\begin{equation*}
\label{E:Kinetic-Energy-ball-wire}
\mathcal{K}= \frac 12 \langle \I \boldsymbol{\Omega}, \boldsymbol{\Omega} \rangle + \frac{m}{2}\left ( (R+r)^2 \dot \vartheta^2 + \dot z^2 
\right ),
\end{equation*}
and, in view of the constraints \eqref{E:Const-ball-on-wire-2}, we obtain the following equations of motion
\begin{eqnarray} \label{E:Motion-ball-wire}
\I \dot {\boldsymbol{\Omega}}&=& \I  \boldsymbol{\Omega} \times  \boldsymbol{\Omega} - \lambda_1 \left ( \frac{R}{R+r}\right )  \boldsymbol{\gamma} + \lambda_2 R\boldsymbol{\beta} \\
m(R+r)^2\ddot \vartheta &=&\lambda_1, \nonumber \\
m\ddot z &=& \lambda_2, \nonumber
\end{eqnarray}
where $\lambda_1, \lambda_2 \in \R$ are Lagrange multipliers. By differentiating the constraints  \eqref{E:Const-ball-on-wire-2}, we obtain the following expressions for the Lagrange multipliers:
\begin{equation*}
\begin{split}
\lambda_1&=m(R+r)^2\ddot \vartheta =mR(R+r)\langle \boldsymbol{\gamma}, \dot{ \boldsymbol{\Omega}} \rangle, \\
\lambda_2&=m\ddot z=\frac{mR^2}{R+r}\langle \boldsymbol{\gamma},{ \boldsymbol{\Omega}} \rangle\langle \boldsymbol{\alpha},{ \boldsymbol{\Omega}} \rangle -mR\langle \boldsymbol{\beta}, \dot{ \boldsymbol{\Omega}} \rangle.
\end{split}
\end{equation*}
Substitution into \eqref{E:Motion-ball-wire} yields
\begin{equation}
\label{E:Motion-ball-wire-2}
(\I + mR^2{\bf I}_3)\dot { \boldsymbol{\Omega}}= \I\boldsymbol{\Omega} \times \boldsymbol{\Omega}
+mR^2\langle \boldsymbol{\alpha}, \dot{ \boldsymbol{\Omega}} \rangle\boldsymbol{\alpha}
+\frac{mR^3}{R+r} \langle \boldsymbol{\gamma},{ \boldsymbol{\Omega}} \rangle \langle \boldsymbol{\alpha},{ \boldsymbol{\Omega}} \rangle  \boldsymbol{\beta},
\end{equation}
where ${\bf I}_3$ denotes the $3\times 3$ identity matrix and we have used the orthonormality
of $ \boldsymbol{\alpha},  \boldsymbol{\beta},  \boldsymbol{\gamma}$. We can solve for
$\dot { \boldsymbol{\Omega}}$ in the above equation 
 by obtaining an  expression for $\langle \boldsymbol{\alpha}, \dot{ \boldsymbol{\Omega}} \rangle$ in terms of $\boldsymbol{\Omega}$ and $ \boldsymbol{\alpha}, \boldsymbol{\beta}$. Taking inner product on both sides of \eqref{E:Motion-ball-wire-2} 
with $(\I+mR^2{\bf I}_3)^{-1}\boldsymbol{\alpha}$ and isolating $\langle \boldsymbol{\alpha}, \dot{ \boldsymbol{\Omega}} \rangle$ gives
\begin{equation*}
\langle \boldsymbol{\alpha}, \dot{ \boldsymbol{\Omega}} \rangle = \frac{\left \langle \I\boldsymbol{\Omega} \times \boldsymbol{\Omega}+\frac{mR^3}{R+r}\langle \boldsymbol{\gamma}, { \boldsymbol{\Omega}} \rangle
\langle \boldsymbol{\alpha}, { \boldsymbol{\Omega}} \rangle \boldsymbol{\beta}\, ,\,  (\I+mR^2{\bf I}_3)^{-1}\boldsymbol{\alpha}\right  \rangle }{1-mR^2\langle (\I+mR^2{\bf I}_3)^{-1}\boldsymbol{\alpha}, \boldsymbol{\alpha} \rangle} .
\end{equation*}
Once the above expression is substituted into \eqref{E:Motion-ball-wire-2} and the system
is complemented with equations \eqref{E:Evolution-Poisson}, we obtain a closed system for
$\boldsymbol{\Omega}, \boldsymbol{\alpha},  \boldsymbol{\beta},  \boldsymbol{\gamma}$ that should
be thought as variables in the reduced space that is obtained by eliminating $\vartheta$ and $z$
by the symmetry.

If the ball is homogeneous with inertia matrix $\I=I\cdot {\bf I}_3$ then $\I\boldsymbol{\Omega} \times \boldsymbol{\Omega}=0$ and also $\langle \boldsymbol{\alpha}, \dot{ \boldsymbol{\Omega}} \rangle=0$ by orthonormality of $\boldsymbol{\alpha}$ and
 $\boldsymbol{\beta}$. The equations simplify to 
 \begin{equation}
 \begin{split}
 \label{E:Motion-ball-wire-homo}
\dot{ \boldsymbol{\Omega} }&= \frac{mR^3}{(R+r)(I+mR^2)}\left \langle \boldsymbol{\gamma}, \boldsymbol{\Omega} \right \rangle  \langle\boldsymbol{\alpha}, \boldsymbol{\Omega} \rangle   \boldsymbol{\beta}, \\
\dot{\boldsymbol{\alpha}}&=\langle \boldsymbol{\gamma}, \boldsymbol{\Omega} \rangle \boldsymbol{\beta} + \boldsymbol{\alpha} \times \boldsymbol{\Omega} ,\\
\dot{\boldsymbol{\beta}}&=-\langle \boldsymbol{\gamma}, \boldsymbol{\Omega} \rangle \boldsymbol{\alpha} + \boldsymbol{\beta} \times \boldsymbol{\Omega} , \\
\dot{\boldsymbol{\gamma}}&= \boldsymbol{\gamma} \times \boldsymbol{\Omega} .
\end{split}
\end{equation}
A direct calculation that uses 
 $\langle \boldsymbol{\gamma}, \boldsymbol{\beta} \rangle =
\langle \boldsymbol{\gamma}, \boldsymbol{\alpha} \rangle=0$,  shows
 that, for any value of $r$  the measure
$d \boldsymbol{\alpha} \wedge d \boldsymbol{\beta} \wedge 
d \boldsymbol{\gamma}\wedge d\boldsymbol{\Omega}$ is preserved by the flow of \eqref{E:Motion-ball-wire-homo}. 
Therefore we have

\begin{theorem}
\label{T:Ball-on-wire}
The reduced equations for a dynamically balanced, axially symmetric sphere that rolls without slipping on  an infinite vertical wire possess an invariant measure if and only if the ball is homogeneous.
\end{theorem}

\section{Conclusions and future work}

We have presented a geometric setup that allows us to obtain necessary and sufficient
conditions for the existence of an invariant measure for certain types of
nonholonomic mechanical systems with symmetry. Our methods have been successfully 
applied to prove the  non-existence of an invariant measure for concrete problems.

Moreover, our geometric framework provides a setup that might be useful to determine
conditions that guarantee the existence of an invariant measure for systems with a particular kind of nonholonomic
constraints (such as   LR systems \cite{FeJo, Veselova}).
 A  class of systems  that we plan to study involves
nonholonomic systems for which the reduced configuration manifold
is a homogeneous space. Some interesting results concerning the existence of a preserved measure for these systems  
have been given in \cite{FeJo, Koi-Rubber}.

This paper has only considered nonholonomic systems with homogeneous constraints. Our results show
that in the absence of a potential, a preserved measure for this kind of systems is necessarily basic (Theorems \ref{T:Basic-Unimodularity} and \ref{T:Main}).  Interestingly, if the nonholonomic constraints are affine, there can exist invariant measures that
are not basic, i.e.,  the density depends on the velocities
\cite{GNMarreroPolacos}.

Finally, it would be interesting to exploit the generalized nonholonomic connection, introduced
in Section \ref{Renomesysy}, in order to study the geometry of symmetric nonholonomic systems which 
don't satisfy, in general, the dimension assumption.

\renewcommand{\thesection}{\sc{Appendix}}

\section{Volume forms on vector bundles}
\label{A:Volume-forms}

\renewcommand{\thesection}{A}

Let  $\tau: E \to Q$ be an orientable vector bundle, over an orientable
 manifold $Q$.
If  $\alpha$ is a section of $\tau: E \to Q$, we can define its vertical lift
$\alpha^{\bf v}$, which is the vector field on $E$ given by 
%
%
\[
\alpha^{\bf v}(\gamma_{q}) = \displaystyle
\frac{d}{dt}_{|t=0}(\gamma_{q} + t\alpha(q)), \mbox{ for }
\gamma_{q} \in E_q.
\]
If $\{e^{\beta}\}$ is a local basis of sections of $E$ and $\alpha = \alpha_{\beta}e^{\beta}$ then
\begin{equation}\label{A:local-vertical-lift}
\alpha^{\bf v} = \displaystyle \alpha_{\beta}\frac{\partial}{\partial p_{\beta}},
\end{equation}
where $p_{\beta}$ are the coordinates on the fibers of $E$ obtained using the basis $\{e^{\beta}\}$.
\begin{lemma}\cite{Marrero}
Let $\nu$ be a volume form on $Q$ and $\Omega$ be a volume form 
on the fibers of $E^*$. Then, there exists a unique
volume form $\nu \wedge \Omega$ on $E$ such that
\begin{equation}\label{Defvol}
\nu \wedge \Omega(\tilde{Z}_{1}, \dots , \tilde{Z}_{m},
\alpha_{1}^{\bf v}, \dots , \alpha_{n}^{\bf v}) = \nu(Z_{1}, \dots
, Z_{m}) \Omega(\alpha_{1}, \dots , \alpha_{n}),
\end{equation}
for $\alpha_{1}, \dots , \alpha_{n} \in \Gamma(\tau)$ and
$\tilde{Z}_{1}, \dots , \tilde{Z}_{m}$ vector fields on $E$
which are $\tau$-projectable on the vector fields $Z_{1},
\dots , Z_{m}$ on $Q$.
\end{lemma}
Locally, if $(q^i)$ are local coordinates on an open subset $U \subseteq Q$ and
$\{e_{\alpha}\}$ is a basis of sections of $E^*$ such that
\[
\nu =  dq^1 \wedge \dots \wedge dq^m, \makebox[.75cm]{} \Omega =
e_{1} \wedge \dots \wedge e_{n},
\] 
then
\begin{equation}
\label{Local-Vol} \nu \wedge \Omega
= dq^1
\wedge \dots \wedge dq^m \wedge dp_{1} \wedge \dots \wedge dp_{n}.
\end{equation}

A volume form $\Phi$ on $E$ is said to be of {\em basic type} if
\begin{equation}
\label{basic-volume}
{\mathcal L}_{\alpha^{\bf v}}\Phi = 0, \; \; \; \forall \alpha \in \Gamma(\tau).
\end{equation}
Using (\ref{Local-Vol}), it is easy to prove that the volume form $\nu \wedge \Omega$ is of basic type. In fact, 
we have the following result
\begin{proposition}
A volume form $\Phi$ on $E$ is of basic type if and only if there exists a volume form $\nu$ on $Q$ 
and a volume form  $\Omega$ on the fibers of $E^*$ such that
\[
\Phi = \nu \wedge \Omega.
\]
\end{proposition}
\begin{proof}
Suppose that $\Phi$ is a volume form on $E$ of basic type.

Let $\nu_0$ be an arbitrary volume form on $Q$ and $\Omega_0$ a volume form on the
fibers of  $E^*$. Then we can assume, without the loss of generality,
that
\[
\Phi = e^{\tilde{\sigma}} \nu_{0} \wedge \Omega_0, \; \; \; \mbox{ with } \tilde{\sigma} \in C^{\infty}(E).
\]
Now, using (\ref{basic-volume}), it follows that
\[
d\tilde{\sigma}(\gamma^{\bf v}) = 0, \; \; \; \forall \gamma \in \Gamma(\tau),
\]
which implies that $\tilde{\sigma}$ is a basic function with respect to the vector bundle projection $\tau: E \to Q$. In
other words, there exists $\sigma \in C^{\infty}(Q)$ such that $\tilde{\sigma} = \sigma \circ \tau$.

Thus, if we take
\[
\nu = e^{\sigma} \nu_0, \; \; \; \Omega = \Omega_0,
\]
we have that $\Phi = \nu \wedge \Omega$.
\end{proof}

\end{document}